\documentclass[a4paper,pdftex,reqno,10pt]{amsart}

\allowdisplaybreaks[1]

\usepackage{geometry}
\usepackage{graphicx}
\usepackage{enumerate}
\usepackage{courier}
\usepackage{times}
\usepackage{amssymb}
\usepackage{amsmath}

\usepackage{hyperref}
\usepackage{url}

\theoremstyle{plain}
\newtheorem{theorem}{Theorem}[section]
\newtheorem{proposition}[theorem]{Proposition}
\newtheorem{conjecture}[theorem]{Conjecture}
\newtheorem{lemma}[theorem]{Lemma}

\theoremstyle{definition}
\newtheorem{definition}[theorem]{Definition}

\newcommand{\upim}{\vec{x}_{\pi_{<i}},\vec{1}_{\pi_{\ge i}}}
\newcommand{\upi}{\vec{x}_{\pi_{\le i}},\vec{1}_{\pi_{>i}}}
\newcommand{\dnim}{\vec{x}_{\pi_{<i}},\vec{0}_{\pi_{\ge i}}}
\newcommand{\dni}{\vec{x}_{\pi_{\le i}},\vec{0}_{\pi_{>i}}}
\newcommand{\upimk}{\vec{x}_{\kappa_{<i}},\vec{1}_{\kappa_{\ge i}}}
\newcommand{\upik}{\vec{x}_{\kappa_{\le i}},\vec{1}_{\kappa_{>i}}}
\newcommand{\dnimk}{\vec{x}_{\kappa_{<i}},\vec{0}_{\kappa_{\ge i}}}
\newcommand{\dnik}{\vec{x}_{\kappa_{\le i}},\vec{0}_{\kappa_{>i}}}
\newcommand{\upimky}{\vec{y}_{\kappa_{<j}},\vec{1}_{\kappa_{\ge j}}}
\newcommand{\upiky}{\vec{y}_{\kappa_{\le j}},\vec{1}_{\kappa_{>j}}}
\newcommand{\dnimky}{\vec{y}_{\kappa_{<j}},\vec{0}_{\kappa_{\ge j}}}
\newcommand{\dniky}{\vec{y}_{\kappa_{\le j}},\vec{0}_{\kappa_{>j}}}
\newcommand{\upimkx}{\vec{x}_{\kappa_{<j}},\vec{1}_{\kappa_{\ge j}}}
\newcommand{\upikx}{\vec{x}_{\kappa_{\le j}},\vec{1}_{\kappa_{>j}}}
\newcommand{\dnimkx}{\vec{x}_{\kappa_{<j}},\vec{0}_{\kappa_{\ge j}}}
\newcommand{\dnikx}{\vec{x}_{\kappa_{\le j}},\vec{0}_{\kappa_{>j}}}
\newcommand{\upjm}{\vec{x}_{\pi_{<j}},\vec{1}_{\pi_{\ge j}}}
\newcommand{\upj}{\vec{x}_{\pi_{\le j}},\vec{1}_{\pi_{>j}}}
\newcommand{\dnjm}{\vec{x}_{\pi_{<j}},\vec{0}_{\pi_{\ge j}}}
\newcommand{\dnj}{\vec{x}_{\pi_{\le j}},\vec{0}_{\pi_{>j}}}

\begin{document}

\title{Importance in systems with interval decisions}

\author{Sascha Kurz}
\address{Department of Mathematics, University of Bayreuth, 95440 Bayreuth, Germany}
\email{sascha.kurz@uni-bayreuth.de}

\date{}

\begin{abstract}
Given a system where the real-valued states of the agents are aggregated by a function to a real-valued 
state of the entire system, we are interested in the influence or importance of the different agents for that function. 
This generalizes the notion of power indices for binary voting systems to decisions over interval 
policy spaces and has applications in economics, engineering, security analysis, and other disciplines. Here, we 
study the question of importance in systems with interval decisions.  
Based on the classical Shapley-Shubik and Penrose-Banzhaf index, from binary voting, we motivate and analyze two importance 
measures. Additionally, we present some results for parametric classes of aggregation functions.

  \medskip
  
  \noindent
  \textbf{Keywords:} Importance; influence; power; interval decisions; state aggregation; Shapley-Shubik index; Penrose-Banzhaf index.\\ 
  \textbf{MSC:} 91B12; 94C10 
\end{abstract}

\maketitle

\section{Introduction}
\label{sec_introduction}
Consider a system where the agents (or components) each determine a real number $x_i$, which is then 
aggregated to another real number $f(\vec{x})$ representing the state of the entire system, where 
$\vec{x}:=(x_1,\dots,x_n)$. This abstract setting occurs in several applications. The values $x_i$ may 
encode the fault condition, normalized between zero and one, for several components of a complex system. 
Then, a suitable aggregation function, see \cite{grabischaggregation} for a survey, might be simply given 
by $f(\vec{x})=\max\{x_1,\dots,x_n\}$. For 
the estimation of unknown quantities, see \cite{galton1907vox}, the {\lq\lq}wisdom of the crowds{\rq\rq}, 
see \cite{surowiecki2004wisdom}, can be applied.\footnote{The effect itself is a purely statistical 
phenomenon and is studied widely in the literature. It can also be simulated by a single individual, see 
\cite{rauhut2011wisdom}. For binary decisions similar effects are studied under the name {\lq\lq}Condorcet 
Jury Theorem{\rq\rq}, see e.g.\ \cite{kanazawa1998brief}.} Depending on the context the mean or the median 
might be a suitable aggregation function, see e.g.\ \cite{chen2004eliminating,galton1907ballot} for some discussion. 
Things get more interesting if the components or agents are heterogeneous in terms of their impact on the 
aggregated value. An example is given by the weighted median, where each agent gets a non-negative integer 
weight $w_i$ such that $\sum_{i=1}^n w_i$ is odd.\footnote{For the definition of the weighted median it is 
neither necessary to restrict to integer weights nor to restrict the possible weight sums of subsets of the 
agents. However, this way we can simplify the technical details cf.~Section~\ref{sec_median} for the more 
general version.} Assume, to ease the notation, that the 
values $x_i$ are pairwise different. Arranging the values in increasing order $x_{i_1}<x_{i_2}<\dots<x_{i_n}$, 
let $1\le j\le n$ be the smallest index such that $\sum_{h=1}^j w_{i_h}>\sum_{h=1}^n w_h/2$. With this, the 
weighted median is given by $x_{i_j}$. Reasons for taking the weighted median instead of the median are 
manifold. When combining the judgment of multiple experts to a single value, different degrees of competence 
may be reflected by different weights, see e.g.\ Chapter~16 in \cite{MayerBooker1991}. In meta analysis, see 
e.g.\ \cite{cooper2009handbook}, aggregated data of differently sized experiments are combined. More generally, 
citing \cite{clark1976effects}: {\lq\lq}The aggregation problem can be defined as the information loss which 
occurs in the substitution of aggregate, or  macrolevel, data for individual, or  microlevel, data.{\rq\rq}. 
Whenever datasets are heterogeneous this has to be reflected somehow in the aggregation, where the weighted 
median is just one possible method, that, however, is commonly applied. For the effects of data aggregation in 
wireless sensor networks we refer e.g.\ to \cite{krishnamachari2002impact}. Even if the micro level data is 
completely available, data aggregation makes sense due to the computational complexity, see e.g.\ 
\cite{xi2009statistical}. Weighted median filters are also applied to sharpen images, see e.g.\ 
\cite{fischer2002weighted}. Due to the increasing share of solar and wind energy, transmission system operators 
are in need of accurate weather forecasts in order to economically regulate the stability of the power grid, see 
e.g.\ \cite{delucchi2011providing}. Typically, several such forecasts are combined with different weights in 
practice. Also fashion retailers invest quite some money to buy more accurate weather forecasts and combine them 
with freely available data, see e.g.\ \cite{bahng2012relationship} for the impact of temperature on sales. 
The median voter model in politics explains the output produced in the public sector by the preferences of the 
median voter, see e.g.\ \cite{holcombe1989median}. While there is some criticism, it is nevertheless applied in 
several applications. Assuming a two-tier voting system, differently sized constituencies of an assembly call for the 
weighted median, see e.g.\ \cite{maaser2007equal}. For a comprehensive and rigor treatise of the general concept 
of aggregation functions we refer the reader to \cite{grabischaggregation}. 

Using the weighted median or another aggregation function, whenever not all agents have an equal impact on the 
aggregation function, the question of the influence or importance of an agent arises. For an example let us continue with the 
weighted median. Assume that we have four agents with weights $w_1=5$, $w_2=4$, $w_3=3$, and $w_4=1$. Here the 
threshold or quota is $7$. Observe that in any ordering of the $x_i$, the value $x_4$ is never the weighted 
median. So, it is justified to say that the fourth agent has no influence on the aggregation function, i.e., no importance at all. For any 
two of the other agents we observe that their weight sum meets or exceeds the quota. So, the second largest value 
of the $x_i$ restricted to the first three agents determines the weighted median. Assuming equal distributions of 
the values $x_i$, we can say that the first three agents are symmetrical and have the same importance. Normalized 
to one, the importance vector of the four agents is given by $(\tfrac{1}{3},\tfrac{1}{3},\tfrac{1}{3},0)$. As a 
by-product we get the information that the weights $\vec{w}=(1,1,1,0)$, with a quota of $2$, lead to the same 
aggregation function when using the weighted median. So, weights can be different from importance. While 
it was easy to determine the importance vector in our example, things get more complicated and even ambiguous in 
more intricate examples like e.g.\ for the weight vector $\vec{w}=(3,1,1,1,1)$.

Nevertheless, the question of determining the importance of an agent or a component in a complex system is very relevant.  
Identifying a component without any importance may allow to remove that component and to reduce production costs. Due  
to security reasons or fault tolerance it might be beneficial if the importance of any component is not too large.  
Important agents can be the goal of bribery or important components be the target of technical attacks. For a firm 
it is important to be not too dependent on one of her external suppliers. From the other side, a supply firm is 
interested in knowing the impact of their contribution to the final product to potentially raise prices. There is also 
another point of view. In an application, the shape of the aggregation function may be defined besides some weights for the 
components, like in the case of the weighted median. Reliability, expertise, accuracy, or any other measure for a 
desired importance vector $\sigma$ given, the question arises how to choose the weights such that the resulting importance vector  
meets $\sigma$ as closely as possible. So, we face a problem of system design.    

The aggregation problem can also be considered as the combination of probability distributions, see e.g,\  
\cite{clemen1999combining,genest1986combining} for an entry point into the related literature.

The question of the importance of agents is studied in the literature for general aggregation functions, see e.g.\ 
\cite[Section 10.3]{grabischaggregation}. However, a huge stream of the literature considers the problem 
restricted to the context of binary voting systems. There, the agents vote 
either {\lq\lq}yes{\rq} or {\lq\lq}no{\rq\rq}, encoded as $1$ and $0$, respectively, on a certain proposal. The aggregated 
group decision then is either to accept (and implement) or to reject the proposal. Von Neumann and Morgenstern introduced the 
notion of a simple game $v$ in \cite{von1953theory}, which is an appropriate model in many applications. For any subset $S$ of supporters 
$v(S)\in \{0,1\}$, where $v$ is surjective and monotone, i.e., $v(S)\le v(T)$ for all $S\subseteq T$. The importance, influence or power 
of an agent in a simple game is measured by so-called power indices like the Shapley-Shubik \cite{shapley1954method} or the 
Penrose-Banzhaf index \cite{banzhaf1964weighted,penrose1946elementary}, see also \cite{felsenthal1998measurement,riker1986first}. The 
model is appropriate to model situations as complex as networks of companies, where several agents own shares of some companies 
that are owning shares of other companies themselves and so are indirectly controlling each other. However, the setting is binary, 
so that economic issues like e.g.\ monetary policy, tax rates, or spending on climate change mitigation does not fit and call 
for an  interval of policy alternatives instead. In the context of voting the system design problem, from the previous paragraph, 
is called {\lq\lq}inverse power index problem{\rq\rq}, see e.g.\ \cite{de2017inverse,koriyama2013optimal,kurz2012inverse}. For 
TU games, a generalization of simple games, the problem is easy, see e.g.\ \cite{dragan2005inverse,faigle2016bases}. For 
non-binary continuous decisions we refer to \cite{kurz2017democratic} and the references therein. Binary decisions with 
continuous signals are e.g.\ considered in \cite{mcmurray2012aggregating}. Even in the binary case 
the importance or power vectors of a given simple game can differ for different power indices, so that the question for the 
{\lq\lq}right{\rq\rq} index arises. Axiomatizations and comparative studies of the properties of the proposed power indices 
aid the practitioner in that task.  
  
Having argued the relevance of the problem, of importance in a complex system with states from an interval, and highlighted its 
connection to voting, we aim to develop importance measures for this setting in the present paper. To this end, we interpret 
the classic Shapley-Shubik and Penrose-Banzhaf indices from a slightly different perspective and generalize the underlying definition 
to our setting. In the same vein the notion of a simple game is generalized. This lays the ground  
to study the question of importance in systems with interval decisions. We remark that some preliminary ideas in that direction have 
been presented in \cite{kurz2014measuring}. While the question is interesting in convex spaces of any dimension, see e.g.\ 
\cite{martin2017owen}, we limit ourselves to intervals of real numbers. The introduction of two measurements of 
importance is not comprehensive at all and more suggestions are deserved. Evaluating the defined importance measures directly 
becomes computationally infeasible quickly if the number of agents increases, which is similar to the situation for power 
indices for simple games. For some classes of aggregation functions we are able to determine either improved algorithms or 
analytical formulas. Also the study of the mathematical properties of the two importance measures is touched.

The remaining part of this paper is structured as follows. In Section~\ref{sec_binary_voting} we briefly collect the basic 
definitions and facts for binary voting systems and power indices. A specific interpretation of the Shapley-Shubik and the 
Penrose-Banzhaf index is the topic of Section~\ref{sec_revisited}. In Section~\ref{sec_model} we generalize simple games 
to simple aggregation functions and power indices to importance measures. Based on the stated interpretation, we generalize 
the Shapley-Shubik and the Penrose-Banzhaf index in Section~\ref{sec_importance_measures}. First mathematical properties 
of these two importance measures are studied in Section~\ref{sec_properties}. Nevertheless, we did not completely succeed 
in revealing the properties of importance in weighted medians, we collect our findings in Section~\ref{sec_median}. We close 
with a conclusion and some open problems in Section~\ref{sec_conclusion}.  

\section{Binary voting systems and power indices} 
\label{sec_binary_voting}

As mentioned in the introduction, we will go by the insights obtained in studying binary voting systems and corresponding 
power indices in order to develop more general importance measures. By $N=\{1,\dots,n\}$ we denote the set of agents. A 
\emph{simple game} is a surjective and monotone mapping $v\colon 2^N\to\{0,1\}$ from the set of subsets $\{S\subseteq N\}$ of $N$, 
i.e., the power set $2^N$ of $N$, into a binary output $\{0,1\}$. \emph{Monotone} means $v(S)\le v(T)$ for all $\emptyset\subseteq 
S\subseteq T\subseteq N$. The values of this mapping can be interpreted as follows. For each subset $S$ of $N$, called \emph{coalition}, 
we have $v(S)=1$ if the members of $S$ can bring through a proposal nevertheless the members of $N\backslash S$ are against it. If $v(S)=1$ 
we speak of a \emph{winning coalition} and a \emph{losing coalition} otherwise. The required monotonicity is quite natural in that 
context, i.e., if the members of a coalition $S$ can bring trough a proposal, then additional supporters should not harm. The 
technical condition of surjectivity, in conjunction with monotonicity, implies that $\emptyset$ is a losing coalition and 
$N$ a winning coalition. This is indeed also quite natural, i.e., if no one supports a proposal then it should not be accepted and 
if otherwise everybody is in favor of a proposal, then there is no reason to reject it. (Typically, surjectivity of $v$ is replaced 
by the equivalent conditions $v(\emptyset)=0$ and $v(N)=1$.) Simple majority for five agents can be modeled by a simple game whose 
winning coalitions are exactly those that have at least three members.

Each simple game is uniquely characterized by either listing all winning or losing coalitions. However, such a representation is not 
very compact.  A slight reduction can be obtained by further exploiting monotonicity. To this end, a winning coalition $S$ is called 
\emph{minimal} if all of its proper subsets are losing. Similarly, a losing coalition $T$ is called \emph{maximal} if all of its proper 
supersets are winning. In our example of simple majority for five agents, the minimal winning coalitions are those with exactly 
three members and the maximal losing coalitions are those with exactly two members. In some cases an even more compact representation, 
based on weights, is possible. Therefore, we call a simple game $v$ weighted if there exist weights $w_1,\dots,w_n\in\mathbb{R}_{\ge 0}$ 
and a quota $q\in\mathbb{R}_{>0}$ such that $v(S)=1$ exactly if $\vec{w}(S):=\sum_{i\in S} w_i\ge q$. As notation we use $[q;\vec{w}]$, i.e., 
$[3;1,1,1,1,1]$ describes simple majority for five agents. As observed in the introduction, different weights can re\-pre\-sent the same 
weighted game, e.g., $[7;5,4,3,1]=[7;4,4,4,1]=[2;1,1,1,0]$. For any weighted game $[q;\vec{w}]$ the difference between the minimum weight 
of a winning and the maximum weight of a losing coalition is some finite positive number, so that we can slightly modify weights and quota 
to rational numbers without changing the underlying simple game. Moreover, by multiplying with the least common multiple of the denominators 
we can assume that the quota and all weights are integers. We note that not every simple game is weighted. However, every simple game 
$v$ can be written as the intersection of a finite number of weighted games $[q_1;\vec{w}_1]$, \dots, $[q_r;\vec{w}_r]$, where 
\begin{equation}
  \left(\bigcap\limits_{i=1}^r [q_i;\vec{w}_i]\right)(S)=\min\left\{[q_i;\vec{w}_i](S)\,:\,1\le i\le r\right\}
\end{equation}    
for all coalitions $S\subseteq N$. For the description of the weighted median in terms of weighted games we need further subclasses of 
simple games. A simple game $v$ is called \emph{proper} if the complement $N\backslash S$ of any winning coalition $S\subseteq N$ is 
losing. If a simple game is not proper, then it may happen that a coalition and its complement can change the status quo by turns, 
which leads to a very unpleasant and unstable situation, so that some researchers only consider proper simple games. Similarly, a 
simple game is called \emph{strong} if the complement $N\backslash T$ of any losing coalition $T\subseteq N$ is winning. A simple 
game that is both proper and strong is called \emph{constant-sum} (or self-dual or decisive). Weighted constant-sum games allow the 
definition of a corresponding aggregation function with a unique weighted median in all cases where the values $x_i$ are pairwise 
different. Integer weights with an odd sum and a quota of half the weight sum (plus one half) are sufficient to guarantee the 
constant-sum property, see Section~\ref{sec_median}. 

Several types of agents can be distinguished in a simple game $v$. Agent~$i\in N$ is called \emph{null} if $v(S)=v(S\cup\{i\})$ 
for all $\emptyset\subseteq S\subseteq N\backslash\{i\}$, i.e., agent~$i$ is not contained in any minimal winning coalition. An agent 
that is contained in every minimal winning coalition is called a veto player. If $\{i\}$ is a winning coalition (note that $\emptyset$ 
is a losing coalition), then player~$i$ is called a \emph{passer}. If additionally all other agents are nulls, then we call 
agent~$i$ a \emph{dictator}. Two agents $i$ and $j$ are called \emph{symmetric}, if $v(S\cup\{i\})=v(S\cup\{j\})$ for all 
$\emptyset\subseteq S\subseteq N\backslash\{i,j\}$. In $[7;5,4,3,1]=[2;1,1,1,0]$ the first three agents are symmetric and the fourth 
agent is a null.

In order to measure the importance of agents in simple games several power indices were introduced in the literature. A power index $p$ 
is a mapping from the set of simple (or weighted) games on $n$ agents into $\mathbb{R}_{\ge 0}^n$. Typically power indices 
are defined for all positive integers $n$, so that we have a family of such mappings. By $p_i(v)$ we denote the $i$th component 
of $p(v)$, i.e., the power of agent~$i$. The \emph{Shapley-Shubik index} is defined as 
\begin{equation}
  \operatorname{SSI}_i(v)=\sum_{S\subseteq N\backslash\{i\}} \frac{|S|!\cdot(n-|S|-1)!}{n!}\cdot\left(v(S\cup\{i\})-v(S)\right).
\end{equation}
We call $S\subseteq N\backslash \{i\}$ a \emph{swing} for agent~$i$ if $v(S\cup\{i\})-v(S)=1$ in a given simple game $v$. In other 
words, $S$ is a losing coalition and $S\cup\{i\}$ a winning coalition. Counting the swings by 
$$\sum_{S\subseteq N\backslash\{i\}}\left(v(S\cup\{i\})-v(S)\right)$$ gives the \emph{absolute Penrose-Banzhaf index}. Normalizing 
via the transformation $p_i(v)/\sum_{j=1}^n p_j(v)$ then gives the \emph{(relative) Penrose-Banzhaf index}. In general, we call 
a power index \emph{efficient} if $\sum_{i=1}^n p_i(v)=1$ for all games $v$.  We call a power index $p$ 
\emph{symmetric} if $p_i(v)=p_j(v)$ for symmetric agents $i,j$ in $v$. If $p_i(v)=0$ for every null $i$ of $v$, then we say that 
$p$ satisfies the \emph{null property}. Both, the Shapley-Shubik and the Penrose-Banzhaf index, are efficient, symmetric, and 
satisfy the null property. The Shapley-Shubik index additionally satisfies the \emph{transfer axiom} 
\begin{equation}
  \varphi_i(u)+\varphi_i(v)= \varphi_i(u\vee v)+\varphi_i(u\wedge v)
\end{equation}    
for all $1\le i\le n$, where $(u\vee v)(S)=\max\{u(S),v(S)\}$ and $(u\wedge v)(S)=\min\{u(S),v(S)\}$. 
In the other direction, the Shapley-Shubik index is the unique power index that satisfies symmetry, efficiency, 
the null property, and the transfer axiom, see \cite{dubey1975uniqueness}. An axiomatization of the Penrose-Banzhaf index 
was given in \cite{dubey1979mathematical}. The absolute Penrose-Banzhaf index also satisfies the transfer axiom.

\section{The definition of the Shapley-Shubik and the Penrose-Banzhaf index revisited}
\label{sec_revisited}

Based on precedent work, the following model was considered in \cite{felsenthal1996alternative}: 
Agents perform a roll-call. More precisely, all $n!$ possible orders $\pi\colon N\to N$ in which
the agents are called are assumed to be equiprobable and the votes of each agent are independent 
with expectation $0\le p\le 1$ for voting $1$, i.e., the probability for voting $1$ is exactly $p$. 
For a given simple game $v$ the \emph{pivotal} agent~$i$ is determined by the unique index $i$ 
such that $\{j\in N\,:\, \pi(j)<\pi(i)\}$ is losing and $\{j\in N\,:\, \pi(j)\le \pi(i)\}$ is winning in $v$. 
Interestingly enough, the Shapley-Shubik index of agent $i$ in $v$ equals the probability that agent $i$ is pivotal 
in the above roll-call model. Note that this statement is independent of $p$. The assumptions on the model 
can be even further weakened to correlated agents still maintaining the coincidence between the Shapley-Shubik index 
and pivot probabilities, see \cite{hu2006asymmetric}.

Let us take another perspective and consider the Shapley-Shubik index as a measurement for the reduction of uncertainty. 
To this end, note that if the votes of all $n$ agents are known, then the aggregated decision, modeled by $v(S)$ for 
the given simple game $v$ and the coalition $S\subseteq N$ of the agents voting {\lq\lq}yes{\rq\rq}, is uniquely determined.  
In the roll-call model we can consider our knowledge on the set of possible outcomes before and after an agent announces 
his or her vote. In the beginning an aggregated decision of both {\lq\lq}yes{\rq\rq} and {\lq\lq}no{\rq\rq} is possible,  
since $v(\emptyset)=0$, $v(N)=1$, and we do not know how the agents will be voting. After the announcement of a certain 
agent the outcome is definitely determined. In other words, that agent reduces the uncertainty about the aggregated outcome 
by one. Let us consider a small example for the simple game $v=[2;1,1,1]$ and the ordering $(1,2,3)$ of the agents. Moreover, 
let us assume that the agents will vote $0$, $0$, and $1$, respectively. After agent~$1$ announces his or her vote both 
outcomes, $1$ or $0$, are possible since the other two agents may both vote $1$ or $0$. After the announcement of agent~$2$, 
the aggregated outcome is determined to be $0$. Also for the ordering $(1,3,2)$ agent~$2$ decides the final outcome. Averaging 
over all possible orderings and voting vectors again gives the Shapley-Shubik index, see \cite{hu2006asymmetric}. 
Note that the aggregated outcome can be determined to either $1$ or $0$, where both cases are symmetric in a certain sense, so 
that the notion of being pivotal applies. If the aggregated decision is a real number in $[0,1]$ instead of $\{0,1\}$ the 
uncertainty about the final outcome can be reduced by several agents at different points in time. Also the degree of reduction 
can be different within the same ordering and input vector $\vec{x}$. So, while there is not much a difference for the binary case, 
those things play a role in the interval case, see Section~\ref{sec_importance_measures}.

In order to formalize things we introduce some more notation. Let $\vec{1}$ and $\vec{0}$ denote the vectors consisting solely of 
ones and zeroes, respectively. For $\vec{x}\in\mathbb{R}^n$ and $S\subseteq\{1,\dots,n\}=N$ we write $\vec{x}_S$ for the restriction 
$(x_i)_{i\in S}$ and $\vec{x}_{-S}$ for $\vec{x}_{N\backslash S}$ (abbreviating $\vec{x}_{-i}=\vec{x}_{-\{i\}}$). 
To each simple game $v$ we associate a mapping $\tilde{v}\colon\{0,1\}^n\to\{0,1\}$, $(\vec{1}_S,\vec{0}_{-S})\mapsto v(S)$. For a 
given permutation $\pi\in\mathcal{S}_n$ of $N$ and $i\in N$, we set $\pi_{<i}=\{j\in N\,:\, \pi(j)<\pi(i)\}$, 
$\pi_{\le i}=\{j\in N\,:\, \pi(j)\le \pi(i)\}$, $\pi_{>i}=N\backslash \pi_{\le i}$, and $\pi_{\ge i}=N\backslash \pi_{< i}$. With this 
and $p=\tfrac{1}{2}$, we have  
$\operatorname{SSI}_i(v) =$
$$
   \frac{1}{n!\cdot 2^n}\cdot\!\!\!\sum_{(\pi,\vec{x})\in\mathcal{S}_n\times \{0,1\}^n}\!\!\! \Big(\tilde{v}(\upim)-\tilde{v}(\dnim)\Big)
    -\Big(\tilde{v}(\upi)-\tilde{v}(\dni))\Big).
$$

Now let us look at the Penrose-Banzhaf index again. Assume, for a given agent~$i$, that all other agents have announced their 
vote. To what degree can agent~$i$ move the aggregated outcome? In the binary setting the range is given by $v(S\cup\{i\})-v(S)$, 
where $v$ is the underlying simple game and $S\subseteq N\backslash\{i\}$ is the set of agents voting {\lq\lq}yes{\rq\rq}. 
If $v(S\cup\{i\})-v(S)=0$, then agent~$i$ cannot change the aggregated outcome at all. If $v(S\cup\{i\})-v(S)=1$, then agent~$i$ 
can shift the aggregated outcome between $0$ and $1$. We may talk of a strategic point of view. In terms of orderings in the roll-call 
model we might say that the Shapley-Shubik index treats all possible orderings equally likely while the Penrose-Banzhaf index 
just considers orderings where the considered agent is last. Again the range of importance is more segmented in the interval case, see 
Section~\ref{sec_importance_measures}. Condensed as a mathematical formula, the Penrose-Banzhaf index for agent~$i$ in a simple game 
$v$ is given by
$$
  \frac{1}{2^n}\cdot \sum_{\vec{x}\in\{0,1\}^n} \Big( \tilde{v}(\vec{x}_{-i},\vec{1}_i)-\tilde{v}(\vec{x}_{-i},\vec{0}_i)\Big).
$$    

\section{The general decision model}
\label{sec_model}

Consider a system of $n$ agents each described by some value $x_i\in\mathcal{I}$. As abbreviation 
we write $\vec{x}$ for the vector $(x_1,\dots,x_n)$, i.e., the state vector of all agents. Further 
assume the existence of an aggregation function $f\colon\mathcal{I}^n\to \mathcal{I}$, where $f(\vec{x})$ 
is a single \textit{aggregated} state. As a crucial 
assumption we require that $\mathcal{I}$ is an interval. In the introduction 
we have mentioned several applications calling for such state spaces. In general the subsequent problem 
is interesting for arbitrary convex spaces. To further ease the notation we consider 
bounded and closed intervals $[a,b]$ only. Via $x\mapsto (x-a)/(b-a)$ we can normalize any such interval 
of positive length to the interval $[0,1]$ that we are considering in the following.     

We are interested in the \textit{importance} of variable $x_i$ for $f$. At a certain state vector $\vec{x}$ 
the partial derivative with respect to variable $x_i$ is an appropriate quantification (assuming differentiability). 
However, we are interested in a more global measure assigning a single non-negative real value to each index $i$, i.e., 
a mapping $p$ from the set of (suitable) aggregation functions into $\mathbb{R}_{\ge 0}^n$. As outlined in the introduction 
this models the importance of a certain agent in a complex system on the decision outcome of the entire system, with the 
aim to distinguish heterogeneous agents according to their degree of importance. As a normalization we require that 
the entries of $p(f)$ sum to one.

With this rather vague description of an importance measure $p$ we remark that, again, this question is interesting for 
a huge variety of aggregation functions. However, in order to obtain stronger results we restrict ourselves on specific 
classes of aggregation functions cf.~\cite{grabischaggregation}. Conducted by the concept of simple games we define, 
cf.~\cite{kurz2014measuring}:

\begin{definition}
  \label{def_simple}
  For a positive integer $n$ a \emph{simple aggregation function} $f$ is a mapping from $[0,1]^n$ to $[0,1]$ that is surjective, 
  continuous, and weakly monotonic increasing, i.e., $f(\vec{x})\le f(\vec{y})$ for all $\vec{x}\le \vec{y}$, i.e., 
  $x_i\le y_i$ for all $1\le i\le n$.  
\end{definition}   

We remark that an \emph{aggregation function}, see e.g.\ \cite{grabischaggregation}, is defined in the literature as a
surjective and weakly monotonic increasing mapping from $\mathcal{I}^n$ to $\mathcal{I}$, where $\mathcal{I}$ is an interval 
of the real numbers. By adding the word \textit{simple} we want to emphasize that $\mathcal{I}=[0,1]$. To ease the subsequent 
mathematical assumptions on the existence of integrals, see e.g.\ Definition~\ref{def_ssi} and Definition~\ref{def_bzi}, and proofs 
we additionally assume continuity.   

For each $\vec{w}\in\mathbb{R}_{\ge 0}^n$ with $\sum_{i=1}^n w_i=1$ the \emph{weighted mean} 
$f(\vec{x}):=\vec{w}^\top \vec{x}=\sum_{i=1}^n w_ix_i$ is a simple aggregation function. For $\vec{w}\in\mathbb{N}^n$, such that 
$\sum_{i=1}^n w_i$ is odd, the \emph{weighted median} is a simple aggregation function. Other examples are given e.g.\ by 
$\hat{f}(x_1,x_2,x_3)=\frac{1x_1^2+2x_2^2+3x_3^2}{6}$ or $\tilde{f}(x_1,x_2,x_3)=x_1x_2^2x_3^3$. Of course there is no need for an explicit 
formula. As an example we consider the so-called \emph{Hegselmann--Krause} or \emph{bounded confidence model} \cite{hegselmann2002opinion}. 
Adjusted to our state space $[0,1]$ the model is as follows: $n$ agents have initial opinions $y_i^0\in[0,1]$. For a given 
parameter $\varepsilon\in (0,\tfrac{1}{2})$ opinions change in discrete time steps via the recursion 
\begin{equation}
  y_i^{t+1}:=\sum_{j\in N_i(t)} y_j^t / \left|N_i(t)\right|,
\end{equation} 
where $N_i(t):=\left\{1\le j\le n\,:\, \left|y_i^t-y_j^t\right|\le\varepsilon\right\}$. Under the stated assumptions, and in 
several generalizations, the opinions converge to a steady state in finite time. Taking $\vec{x}$ as the initial opinions 
we may define $f(\vec{x})$ as the resulting steady state and easily check that this also gives a simple aggregation function 
(except possibly continuity). 
Since all agents (or variables) are homogeneous the question of importance is not very interesting. However, we can 
simply generalize the model and include some weights. Opinion dynamics based on further ways of averaging are studied in 
\cite{hegselmann2005opinion} and may serve for the definition of interesting simple aggregation functions. 

Having the transfer axiom in mind, we easily observe:
\begin{lemma}
  Let $f$ and $g$ be two simple aggregation functions for the same number $n\ge 1$ of agents. Then, 
  $f\vee g$ and $f\wedge g$, defined by $(f\vee g)(\vec{x})=\max\{f(\vec{x}),g(\vec{x})\}$ and 
  $(f\wedge g)(\vec{x})=\min\{f(\vec{x}),g(\vec{x})\}$ for all $\vec{x}\in[0,1]^n$, are simple aggregation 
  functions.   
\end{lemma}    
 
We also want to transfer the classification of types of agents in a simple aggregation function. 
\begin{definition}
  \label{def_agent_types}
  Let $f$ be a simple aggregation function for $n$ agents.
  \begin{enumerate}
    \item[(i)] If $f(\vec{x})=f(\vec{y})$ for all $\vec{x},\vec{y}\in[0,1]^n$ with $x_j=y_j$ for all $j\in \{1,\dots,n\}\backslash\{i\}$, 
               then we call agent~$i$ a \emph{null}.
    \item[(ii)] If $f(\vec{x})=f(\vec{y})$ for all $\vec{x},\vec{y}\in[0,1]^n$ with $x_i=y_i$, 
               then we call agent~$i$ a \emph{dictator}.
    \item[(iii)] If $f(\vec{x})=f(\vec{y})$ for all $\vec{x},\vec{y}\in[0,1]^n$ with $x_i=y_j$, $x_j=y_i$, and $x_h=y_h$ for all 
               $h\in \{1,\dots,n\}\backslash\{i,j\}$, then we call agent~$i$ and agent~$j$ \emph{equivalent}.                      
  \end{enumerate}
\end{definition}

It seems reasonable to require the following conditions for a measure of importance in a simple aggregation function:
\begin{definition}
  \label{def_importance_measure}
  An \emph{importance measure} is a mapping $p$ from the set of all simple aggregation functions for $n$ agents into $\mathbb{R}_{\ge 0}^n$ 
  that is
  \begin{enumerate}
    \item[(i)] \emph{efficient}, i.e., $\sum_{i=1}^n p_i(f)=1$;
    \item[(ii)] \emph{symmetric}, i.e., $p_i(f)=p_j(f)$ for symmetric agents $1\le i,j\le n$; and
    \item[(iii)] has the \emph{null property}, i.e., $p_i(f)=0$ for every null $1\le i\le n$.    
  \end{enumerate} 
\end{definition}   

Note that efficiency and the null property implies $p_i(f)=1$ for an agent~$i$ that is a dictator, since 
all other agents have to be nulls.  

\begin{definition}
  An \emph{importance measure} $p$ satisfies the \emph{transfer axiom} if
  $p_i(f)+p_i(g)= p_i(f\vee g)+p_i(f\wedge g)$   
  for all simple aggregation functions $f$ and $g$ on $n\ge 1$ agents and all agents $1\le i\le n$.
\end{definition}

Depending on the application further 
properties might be desirable. In the following section we will introduce two reasonable importance measures.

\section{Two importance measures for systems with interval decisions}
\label{sec_importance_measures}
Motivated by the interpretation of the Shapley-Shubik and the Penrose-Banzhaf index in Section~\ref{sec_revisited}, we introduce 
two importance measures for simple aggregation functions. Let us start with the generalization of the Shapley-Shubik index. We stick 
to the roll-call model and assume a given ordering $\pi\colon N\to\ N$ of the agents and a given simple aggregation function $f$. For 
a given agent $i$ we consider the case where all agents with $\pi(j)<\pi(i)$ have already announced there state $x_j$. Given that information 
we are interested in the uncertainty of the possible value of $f(\vec{x})$, where $\vec{x}$ is only partially specified. 

Since $f$ is monotone, $f(\upim)$ is the maximal value that can be attained by $f(\vec{x})$ if the $x_j$ of all 
agents with $\pi(j)<\pi(i)$ are fixed. Similarly, $f(\dnim)$ is the minimal value that can be attained by 
$f(\vec{x})$ if the $x_j$ of all agents with $\pi(j)<\pi(i)$ are fixed. Since $f$ is continuous, all values in the interval between 
$f(\dnim)$ and $f(\upim)$ can be attained by some vector $\vec{x}\in[0,1]^n$, 
where the entries of all agents with $\pi(j)<\pi(i)$ are fixed. The length
$
  f(\upim)-f(\dnim)
$
of that interval is a suitable measure for the uncertainty of the simple aggregation function $f$ before agent $i$ announces his or her state $x_i$, 
with respect to the ordering $\pi$ and the state vector $\vec{x}$. Similarly, the uncertainty after the announcement of agent~$i$ is given 
by 
$
  f(\upi)-f(\dni)
$.
The difference between both values 
summed over all possible orderings and averaged over all possible state vectors $\vec{x}$ may serve as a suitable 
measurement of importance:
\begin{definition}
  \label{def_ssi}
    For a simple aggregation function $f$ for $n\ge 1$ agents we set
  \begin{eqnarray}
    \varphi_i(f)&:=&\frac{1}{n!}\cdot\sum_{\pi\in\mathcal{S}_n} \int_0^1\dots\int_0^1
    \Big(f(\upim)-f(\dnim)\Big)\nonumber\\
    &&-\Big(f(\upi)-f(\dni)\Big) 
    \operatorname{d}x_1\,\dots\,\operatorname{d}x_n,
  \end{eqnarray}
  for each agent $i\in N$. 
\end{definition}
Here we assume that the states of all agents are independent and that all state vectors $\vec{x}$ occur with equal probability. This 
assumption can of course be adjusted easily. As an example we consider the two simple aggregation functions 
$\hat{f}(x_1,x_2,x_3)=\frac{1x_1^2+2x_2^2+3x_3^2}{6}$ and $\tilde{f}(x_1,x_2,x_3)=x_1x_2^2x_3^3$. 

\begin{table}[htp]
  \begin{center}
    \begin{tabular}{cc}
      \hline
      $\pi\in\mathcal{S}_3$ & $3$-fold integral\\
      \hline
      $(1,2,3)$ & $\int\limits_{\vec{x}\in[0,1]^3} \left(\frac{6}{6}-\frac{x_1^2+5}{6}+\frac{x_1^2}{6}-\frac{0}{6}\right)\,\operatorname{d}\vec{x}=\frac{1}{6}$ \\
      $(1,3,2)$ & $\int\limits_{\vec{x}\in[0,1]^3} \left(\frac{6}{6}-\frac{x_1^2+5}{6}+\frac{x_1^2}{6}-\frac{0}{6}\right)\,\operatorname{d}\vec{x}=\frac{1}{6}$ \\
      $(2,1,3)$ & $\int\limits_{\vec{x}\in[0,1]^3} \left(\frac{2x_2^2+4}{6}-\frac{x_1^2+2x_2^2+3}{6}+\frac{x_1^2+2x_2^2}{6}-\frac{2x_2^2}{6}\right)\,\operatorname{d}\vec{x}=\frac{1}{6}$ \\
      $(2,3,1)$ & $\int\limits_{\vec{x}\in[0,1]^3} \left(\frac{2x_2^2+3x_3^2+1}{6}-\frac{x_1^2+2x_2^2+3x_3^2}{6}+\frac{x_1^2+2x_2^2+3x_3^2}{6}-\frac{2x_2^2+3x_3^2}{6}\right)\,\operatorname{d}\vec{x}=\frac{1}{6}$ \\
      $(3,1,2)$ & $\int\limits_{\vec{x}\in[0,1]^3} \left(\frac{3x_3^2+3}{6}-\frac{x_1^2+3x_3^2+2}{6}+\frac{x_1^2+3x_3^2}{6}-\frac{3x_3^2}{6}\right)\,\operatorname{d}\vec{x}=\frac{1}{6}$ \\
      $(3,2,1)$ & $\int\limits_{\vec{x}\in[0,1]^3} \left(\frac{2x_2^2+3x_3^2+1}{6}-\frac{x_1^2+2x_2^2+3x_3^2}{6}+\frac{x_1^2+2x_2^2+3x_3^2}{6}-\frac{2x_2^2+3x_3^2}{6}\right)\,\operatorname{d}\vec{x}=\frac{1}{6}$ \\
      \hline
    \end{tabular}
    \caption{Details for $\varphi_1(\hat{f})$.}
    \label{table_ex_1_ssi_1}
  \end{center}
\end{table}

\begin{table}[htp]
  \begin{center}
    \begin{tabular}{cc}
      \hline
      $\pi\in\mathcal{S}_3$ & $3$-fold integral\\
      \hline
      $(2,1,3)$ & $\int\limits_{\vec{x}\in[0,1]^3} \left(\frac{6}{6}-\frac{x_2^2+4}{6}+\frac{x_2^2}{6}-\frac{0}{6}\right)\,\operatorname{d}\vec{x}=\frac{2}{6}$ \\
      $(2,3,1)$ & $\int\limits_{\vec{x}\in[0,1]^3} \left(\frac{6}{6}-\frac{x_2^2+4}{6}+\frac{x_2^2}{6}-\frac{0}{6}\right)\,\operatorname{d}\vec{x}=\frac{2}{6}$ \\
      $(1,2,3)$ & $\int\limits_{\vec{x}\in[0,1]^3} \left(\frac{x_1^2+5}{6}-\frac{x_1^2+2x_2^2+3}{6}+\frac{x_1^2+2x_2^2}{6}-\frac{x_1^2}{6}\right)\,\operatorname{d}\vec{x}=\frac{2}{6}$ \\
      $(3,2,1)$ & $\int\limits_{\vec{x}\in[0,1]^3} \left(\frac{3x_3^2+3}{6}-\frac{2x_2^2+3x_3^2+1}{6}+\frac{2x_2^2+3x_3^2}{6}-\frac{3x_3^2}{6}\right)\,\operatorname{d}\vec{x}=\frac{2}{6}$ \\
      $(1,3,2)$ & $\int\limits_{\vec{x}\in[0,1]^3} \left(\frac{x_1^2+3x_3^2+2}{6}-\frac{x_1^2+2x_2^2+3x_3^2}{6}+\frac{x_1^2+2x_2^2+3x_3^2}{6}-\frac{x_1^2+3x_3^2}{6}\right)\,\operatorname{d}\vec{x}=\frac{2}{6}$ \\
      $(3,1,2)$ & $\int\limits_{\vec{x}\in[0,1]^3} \left(\frac{x_1^2+3x_3^2+2}{6}-\frac{x_1^2+2x_2^2+3x_3^2}{6}+\frac{x_1^2+2x_2^2+3x_3^2}{6}-\frac{x_1^2+3x_3^2}{6}\right)\,\operatorname{d}\vec{x}=\frac{2}{6}$ \\
      \hline
    \end{tabular}
    \caption{Details for $\varphi_2(\hat{f})$.}
    \label{table_ex_1_ssi_2}
  \end{center}
\end{table}

\begin{table}[htp]
  \begin{center}
    \begin{tabular}{cc}
      \hline
      $\pi\in\mathcal{S}_3$ & $3$-fold integral\\
      \hline
      $(3,1,2)$ & $\int\limits_{\vec{x}\in[0,1]^3} \left(\frac{6}{6}-\frac{3x_3^2+3}{6}+\frac{3x_3^2}{6}-\frac{0}{6}\right)\,\operatorname{d}\vec{x}=\frac{3}{6}$ \\
      $(3,2,1)$ & $\int\limits_{\vec{x}\in[0,1]^3} \left(\frac{6}{6}-\frac{3x_3^2+3}{6}+\frac{3x_3^2}{6}-\frac{0}{6}\right)\,\operatorname{d}\vec{x}=\frac{3}{6}$ \\
      $(1,3,2)$ & $\int\limits_{\vec{x}\in[0,1]^3} \left(\frac{x_1^2+5}{6}-\frac{x_1^2+3x_3^2+2}{6}+\frac{x_1^2+3x_3^2}{6}-\frac{x_1^2}{6}\right)\,\operatorname{d}\vec{x}=\frac{3}{6}$ \\
      $(2,3,1)$ & $\int\limits_{\vec{x}\in[0,1]^3} \left(\frac{2x_2^2+4}{6}-\frac{2x_2^2+3x_3^2+1}{6}+\frac{2x_2^2+3x_3^2}{6}-\frac{2x_2^2}{6}\right)\,\operatorname{d}\vec{x}=\frac{3}{6}$ \\
      $(1,2,3)$ & $\int\limits_{\vec{x}\in[0,1]^3} \left(\frac{x_1^2+2x_2^2+3}{6}-\frac{x_1^2+2x_2^2+3x_3^2}{6}+\frac{x_1^2+2x_2^2+3x_3^2}{6}-\frac{x_1^2+2x_2^2}{6}\right)\,\operatorname{d}\vec{x}=\frac{3}{6}$ \\
      $(2,1,3)$ & $\int\limits_{\vec{x}\in[0,1]^3} \left(\frac{x_1^2+2x_2^2+3}{6}-\frac{x_1^2+2x_2^2+3x_3^2}{6}+\frac{x_1^2+2x_2^2+3x_3^2}{6}-\frac{x_1^2+2x_2^2}{6}\right)\,\operatorname{d}\vec{x}=\frac{3}{6}$ \\
      \hline
    \end{tabular}
    \caption{Details for $\varphi_3(\hat{f})$.}
    \label{table_ex_1_ssi_3}
  \end{center}
\end{table}
 
In tables~\ref{table_ex_1_ssi_1}-\ref{table_ex_1_ssi_3} we give the respective summands for each permutation $\pi\in\mathcal{S}_3$ for 
$\hat{f}$. Summarizing the results we obtain 
\begin{equation}
  \varphi(\hat{f})=\left(\frac{1}{6},\frac{2}{6},\frac{3}{6}\right).
\end{equation}   
Note that $\varphi$ is efficient in our example, which we will prove in general in Lemma~\ref{lemma_ssi_properties}. 
Moreover, the entries of $\varphi(\hat{f})$ coincide with the coefficients of the linear combination of functions defining $\hat{f}$. 
Again, this is a general phenomenon, see Theorem~\ref{thm_linear_combination}.

\begin{table}[htp]
  \begin{center}
    \begin{tabular}{cc}
      \hline
      $\pi\in\mathcal{S}_3$ & $3$-fold integral\\
      \hline
      $(1,2,3)$ & $\int\limits_{\vec{x}\in[0,1]^3} \left(1-x_1+0-0\right)\,\operatorname{d}\vec{x}=\frac{1}{2}$ \\
      $(1,3,2)$ & $\int\limits_{\vec{x}\in[0,1]^3} \left(1-x_1+0-0\right)\,\operatorname{d}\vec{x}=\frac{1}{2}$ \\
      $(2,1,3)$ & $\int\limits_{\vec{x}\in[0,1]^3} \left(x_2^2-x_1x_2^2+0-0\right)\,\operatorname{d}\vec{x}=\frac{1}{6}$ \\
      $(2,3,1)$ & $\int\limits_{\vec{x}\in[0,1]^3} \left(x_3^3-x_1x_3^3+0-0\right)\,\operatorname{d}\vec{x}=\frac{1}{8}$ \\
      $(3,1,2)$ & $\int\limits_{\vec{x}\in[0,1]^3} \left(x_2^2x_3^3-x_1x_2^2x_3^3+x_1x_2^2x_3^3-0\right)\,\operatorname{d}\vec{x}=\frac{1}{12}$ \\
      $(3,2,1)$ & $\int\limits_{\vec{x}\in[0,1]^3} \left(x_2^2x_3^3-x_1x_2^2x_3^3+x_1x_2^2x_3^3-0\right)\,\operatorname{d}\vec{x}=\frac{1}{12}$ \\
      \hline
    \end{tabular}
    \caption{Details for $\varphi_1(\tilde{f})$.}
    \label{table_ex_2_ssi_1}
  \end{center}
\end{table}

\begin{table}[htp]
  \begin{center}
    \begin{tabular}{cc}
      \hline
      $\pi\in\mathcal{S}_3$ & $3$-fold integral\\
      \hline
      $(1,2,3)$ & $\int\limits_{\vec{x}\in[0,1]^3} \left(x_1-x_1x_2^2+0-0\right)\,\operatorname{d}\vec{x}=\frac{1}{3}$ \\
      $(1,3,2)$ & $\int\limits_{\vec{x}\in[0,1]^3} \left(x_1x_3^3-x_1x_2^2x_3^3+x_1x_2^2x_3^3-0\right)\,\operatorname{d}\vec{x}=\frac{1}{8}$ \\
      $(2,1,3)$ & $\int\limits_{\vec{x}\in[0,1]^3} \left(1-x_2^2+0-0\right)\,\operatorname{d}\vec{x}=\frac{2}{3}$ \\
      $(2,3,1)$ & $\int\limits_{\vec{x}\in[0,1]^3} \left(1-x_2^2+0-0\right)\,\operatorname{d}\vec{x}=\frac{2}{3}$ \\
      $(3,1,2)$ & $\int\limits_{\vec{x}\in[0,1]^3} \left(x_1x_3^3-x_1x_2^2x_3^3+x_1x_2^2x_3^3-0\right)\,\operatorname{d}\vec{x}=\frac{1}{8}$ \\
      $(3,2,1)$ & $\int\limits_{\vec{x}\in[0,1]^3} \left(x_3^3-x_2^2x_3^3+0-0\right)\,\operatorname{d}\vec{x}=\frac{1}{6}$ \\
      \hline
    \end{tabular}
    \caption{Details for $\varphi_2(\tilde{f})$.}
    \label{table_ex_2_ssi_2}
  \end{center}
\end{table}

\begin{table}[htp]
  \begin{center}
    \begin{tabular}{cc}
      \hline
      $\pi\in\mathcal{S}_3$ & $3$-fold integral\\
      \hline
      $(1,2,3)$ & $\int\limits_{\vec{x}\in[0,1]^3} \left(x_1x_2^2-x_1x_2^2x_3^3+x_1x_2^2x_3^3-0\right)\,\operatorname{d}\vec{x}=\frac{1}{6}$ \\
      $(1,3,2)$ & $\int\limits_{\vec{x}\in[0,1]^3} \left(x_1-x_1x_3^3+0-0\right)\,\operatorname{d}\vec{x}=\frac{3}{8}$ \\
      $(2,1,3)$ & $\int\limits_{\vec{x}\in[0,1]^3} \left(x_1x_2^2-x_1x_2^2x_3^3+x_1x_2^2x_3^3-0\right)\,\operatorname{d}\vec{x}=\frac{1}{6}$ \\
      $(2,3,1)$ & $\int\limits_{\vec{x}\in[0,1]^3} \left(x_2^2-x_2^2x_3^3+0-0\right)\,\operatorname{d}\vec{x}=\frac{1}{4}$ \\
      $(3,1,2)$ & $\int\limits_{\vec{x}\in[0,1]^3} \left(1-x_3^3+0-0\right)\,\operatorname{d}\vec{x}=\frac{3}{4}$ \\
      $(3,2,1)$ & $\int\limits_{\vec{x}\in[0,1]^3} \left(1-x_3^3+0-0\right)\,\operatorname{d}\vec{x}=\frac{3}{4}$ \\      
      \hline
    \end{tabular}
    \caption{Details for $\varphi_3(\tilde{f})$.}
    \label{table_ex_2_ssi_3}
  \end{center}
\end{table}

In tables~\ref{table_ex_2_ssi_1}-\ref{table_ex_2_ssi_3} we give the respective summands for each permutation
$\pi\in\mathcal{S}_3$ for $\tilde{f}$. Summarizing the results we obtain 
\begin{equation}
  \varphi(\tilde{f})=\left(\frac{35}{144},\frac{50}{144},\frac{59}{144}\right)
  =\left(0.2430\overline{5},0.347\overline{2},0.4097\overline{2}\right).
\end{equation}
In Theorem~\ref{thm_power_product} we state a more general formula for $\varphi\left(x_1^{\alpha_1}x_2^{\alpha_2}\dots x_n^{\alpha_n}\right)$
where $\alpha_i\in\mathbb{R}_{>0}$.

For the generalization of the Penrose-Banzhaf index we stick to the strategic point of view outlined in Section~\ref{sec_revisited}. 
So, for a given simple aggregation function $v$ and a state vector $\vec{x}$ agent~$i$ increases the value of $f(\vec{x})$ maximally by  
choosing $x_i=1$, due to monotonicity of $f$. Similarly, the minimum is attained for $x_i=0$.
\begin{definition}
  \label{def_bzi} (Cf.\ Definition 10.29 in \cite{grabischaggregation})
  For a simple aggregation function $f$ for $n\ge 2$ agents we set
  \begin{eqnarray}
    \widetilde{\psi}_i(f)&:=& \int_0^1\dots\int_0^1 \Big(f(\vec{x}_{-i},\vec{1}_{i})
    -f(\vec{x}_{-i},\vec{0}_{i})\Big)\operatorname{d}x_1\dots\operatorname{d}x_n\\
    &=& \int_0^1\dots\int_0^1 \Big(f(\vec{x}_{-i},\vec{1}_{i})
    -f(\vec{x}_{-i},\vec{0}_{i})\Big)\operatorname{d}x_1\dots\operatorname{d}x_{i-1}\, \operatorname{d}x_{i+1}\dots\operatorname{d}x_n\nonumber.
  \end{eqnarray}
  for each agent $i\in N$. With this, we normalize to $\psi_i(f)=\widetilde{\psi}_i(f) / \sum_{j=1}^n \widetilde{\psi}_j(f)$. 
\end{definition}
Here we again assume that all state vectors $\vec{x}$ occur with equal probability, which can of course be adjusted easily. 
As an example we consider the same two simple aggregation functions 
$\hat{f}(x_1,x_2,x_3)=\frac{1x_1^2+2x_2^2+3x_3^2}{6}$ and $\tilde{f}(x_1,x_2,x_3)=x_1x_2^2x_3^3$ as before. Here, we obtain:
\begin{eqnarray}
  \widetilde{\psi}_1(\hat{f}) &=& \int_0^1\int_0^1 \left(\frac{1+2x_2^2+3x_3^3}{6}-\frac{0+2x_2^2+3x_3^3}{6}\right)\operatorname{d}x_2\,\operatorname{d}x_3=\frac{1}{6},\nonumber\\
  \widetilde{\psi}_2(\hat{f}) &=& \int_0^1\int_0^1 \left(\frac{1x_1+2+3x_3^3}{6}-\frac{1x_1^1+0+3x_3^3}{6}\right)\operatorname{d}x_1\,\operatorname{d}x_3=\frac{2}{6},\nonumber\\
  \widetilde{\psi}_3(\hat{f}) &=& \int_0^1\int_0^1 \left(\frac{1x_1^2+2x_2^2+3}{6}-\frac{1x_1^2+2x_2^2+0}{6}\right)\operatorname{d}x_1\,\operatorname{d}x_2=\frac{3}{6}.
\end{eqnarray}
Since $\widetilde{\psi}_1(\hat{f})+\widetilde{\psi}_2(\hat{f})+\widetilde{\psi}_3(\hat{f})=1$ no normalization
is necessary. Moreover, the example is covered by Theorem~\ref{thm_linear_combination}. For the other example we obtain:
\begin{eqnarray}
  \widetilde{\psi}_1(\tilde{f}) &=& \int_0^1\int_0^1 \left(x_2^2x_3^3-0\right)\operatorname{d}x_2\,\operatorname{d}x_3=\frac{1}{12},\nonumber\\
  \widetilde{\psi}_2(\tilde{f}) &=& \int_0^1\int_0^1 \left(x_1x_3^3-0\right)\operatorname{d}x_1\,\operatorname{d}x_3=\frac{1}{8},\nonumber\\
  \widetilde{\psi}_3(\tilde{f}) &=& \int_0^1\int_0^1 \left(x_1x_2^2-0\right)\operatorname{d}x_1\,\operatorname{d}x_2=\frac{1}{6}.
\end{eqnarray}
After normalization we obtain $\psi=\frac{1}{9}\cdot\left(2,3,4\right)=\left(0.\overline{2},0.\overline{3},0.\overline{4}\right)$. In 
Theorem~\ref{thm_power_product} we capture this example as a special case.
 
We remark that several results for the importance measure from Definition~\ref{def_bzi} are known, see e.g.\ 
\cite[Section~10.3]{grabischaggregation}. The \emph{dual aggregation function} $f^d(\vec{x}):=1-f(1-x_1,\dots,1-x_n)$ 
satisfies $\widetilde{\psi}_i(f)=\widetilde{\psi}_i(f^d)$ and $\psi_i(f)=\psi_i(f^d)$ for all $1\le i\le n$, see e.g.\ 
\cite[Proposition~10.30]{grabischaggregation}. 
Starting from a simple game $v$, in Section~\ref{sec_revisited} we have written the Shapley-Shubik and the Penrose-Banzhaf index in the 
shape of Definition~\ref{def_ssi} and Definition~\ref{def_bzi}, respectively. However, the associated function $\tilde{v}$ is not a 
simple aggregation function since it is only defined for the domain $\{0,1\}^n$. If we apply the Choquet integral to $v$ we obtain 
an aggregation function $\hat{v}$ such that $\psi(\hat{v})$ equals the Shapley-Shubik index of $v$, see 
\cite[Proposition 10.31]{grabischaggregation}. If the Choquet integral is replaced with the multilinear extension of Owen \cite{owen1972multilinear} 
and plugged into $\psi$, then we end up with the Penrose-Banzhaf index of $v$, see \cite[Proposition 10.34]{grabischaggregation}.

\section{Properties of the two importance measures}
\label{sec_properties}
First, we have to verify that the two mappings $\varphi$ and $\psi$ (see Definition~\ref{def_ssi} and Definition~\ref{def_bzi}) are indeed 
importance measures, i.e., that they satisfy the conditions of Definition~\ref{def_importance_measure}. 

\begin{lemma}
  \label{lemma_ssi_properties}
  For a positive integer $n$ the mapping $\varphi$ is a well-defined importance measure that satisfies the transfer axiom.
\end{lemma}
\begin{proof}
  Every argument $f$ of $\varphi$ is a simple aggregation function, which especially means that $f$ is continuous 
  over the compact domain $[0,1]^n$. Thus, the integrals in the definition of $\varphi(f)$ exist, so that the mapping 
  $\varphi$ is well-defined. Moreover, $f(\upim)\ge f(\upi)$ 
  and $f(\dni)\ge f(\dnim)$ for every $\pi\in\mathcal{S}_n$, 
  $\vec{x}\in [0,1]^n$, and $i\in N$, so that $\varphi_i(f)\ge 0$.
  
  For any permutation $\pi\in\mathcal{S}_n$ and any $0\le h\le n$ let $\pi|h:=\{\pi(i)\,:\,1\le i\le h\}$, i.e., 
  the first $h$ agents in ordering $\pi$. 
  Then, for any state vector $\vec{x}\in[0,1]^n$, we have
  \begin{eqnarray}
    &&\sum_{i=1}^n f(\upim)-f(\upi)+f(\dni)-f(\dnim)\nonumber\\
    &=&\sum_{h=1}^n f(\vec{x}_{\pi|h-1},\vec{1}_{-\pi|h-1})\!-\! f(\vec{x}_{\pi|h},\vec{1}_{-\pi|h})
    \!+\!\sum_{h=1}^n f(\vec{x}_{\pi|h},\vec{0}_{-\pi|h})\!-\!f(\vec{x}_{\pi|h-1},\vec{0}_{-\pi|h-1})\nonumber\\
    &=& f(\vec{x}_{\pi|0},\vec{1}_{-\pi|0})-f(\vec{x}_{\pi|n},\vec{1}_{-\pi|n})+f(\vec{x}_{\pi|n},\vec{0}_{-\pi|n})-
    f(\vec{x}_{\pi|0},\vec{0}_{-\pi|0})\nonumber\\
    &=& f(\vec{1})-f(\vec{x})+f(\vec{x})-f(\vec{0})=1-0=1 
  \end{eqnarray}
  so that $\sum_{i=1}^n \varphi_i(f)=1$, i.e., $\varphi$ is efficient.
  
  For two distinct symmetric agents $i,j\in N$ let $\sigma\in\mathcal{S}_n$ be the transposition interchanging 
  agent~$i$ and agent~$j$. For a given state vector $\vec{x}\in[0,1]^n$ we define $\vec{y}\in[0,1]^n$ as the 
  vector arising from interchanging the $i$th and the $j$th coordinate of $\vec{x}$. By $\kappa\in\mathcal{S}_n$ 
  we denote the concatenation of $\sigma$ with $\pi$. With this, we have
  \begin{eqnarray*}
  &&\sum_{\pi\in\mathcal{S}_n} f(\upim)-f(\upi)+f(\dni)-f(\dnim)\\
  &=&\sum_{\kappa\in\mathcal{S}_n} f(\upimk)-f(\upik)+f(\dnik)-f(\dnimk)\\
  &=&\sum_{\kappa\in\mathcal{S}_n} f(\upimky)-f(\upiky)+f(\dniky)-f(\dnimky)\\
  &=&\sum_{\kappa\in\mathcal{S}_n} f(\upimkx)-f(\upikx)+f(\dnikx)-f(\dnimkx)\\
  &=&\sum_{\pi\in\mathcal{S}_n} f(\upjm)-f(\upj)+f(\dnj)-f(\dnjm)
  \end{eqnarray*}  
  so that $\varphi_i(f)=\varphi_j(f)$, i.e., $\varphi$ is symmetric.
  
  If agent $i\in N$ is a null and $\pi\in\mathcal{S}_n$ arbitrary, then 
  $f(\dnim)=f(\dni)$ and
  $f(\upim)=f(\upi)$, so that $\varphi_i(f)=0$, 
  i.e., $\varphi$ satisfies the null property.
  
  Since $x+y=\max\{x,y\}+\min\{x,y\}$ for all $x,y\in\mathbb{R}$ and due to the linearity of finite sums and integrals, $\varphi$ satisfies 
  the transfer axiom.
\end{proof}

\begin{lemma}
  \label{lemma_bzi_properties}
  For a positive integer $n\ge 2$ the mapping $\psi$ is a well-defined importance measure. Moreover, $\widetilde{\psi}$ satisfies 
  the transfer axiom.
\end{lemma}
\begin{proof}
  Every argument $f$ of $\widetilde{\psi}$ and $\psi$ is a simple aggregation function, which especially means that $f$ 
  is continuous over the compact domain $[0,1]^n$. Thus, the integrals in the definition of $\widetilde{\psi}(f)$ exist, 
  so that the mapping $\widetilde{\psi}$ is well-defined. Since $f(\vec{x}_{-i},\vec{1}_{i})\ge 
  f(\vec{x}_{-i},\vec{0}_{i})$ for every $i\in N$, due to monotonicity, we have $\widetilde{\psi}_i(f)\ge 0$.
  
  Since $f$ is weakly monotonic increasing, $f(\vec{0})=0$, and $f(\vec{1})=1$ there exists a state vector $\vec{x}\in[0,1]^n$ 
  and an agent~$i\in N$ such that $$\varepsilon:=f(\vec{x}_{-i},\vec{1}_{i})-f(\vec{x}_{-i},\vec{0}_{i})>0.$$  
  Due to continuity, there exists a constant $0<\delta<\tfrac{1}{2}$ such that 
  $f(\vec{x'}_{-i},\vec{1}_{i})-f(\vec{x'}_{-i},\vec{0}_{i})\ge\varepsilon/2$ for all 
  $l_h:=\max\{0,x_h-\delta\}\le x'_h\le \min\{1,x_h+\delta\}=:u_h$ and all $h\in N\backslash\{i\}$. Since $u_h-l_h\ge\delta$ we have
  $$
    \widetilde{\psi}_i(f)\ge \int_{l_0}^{u_1}\dots\int_{l_{i-1}}^{u_{i-1}}\int_{l_{i+1}}^{u_{i+1}}\dots\int_{l_{n}}^{u_{n}} 
    \Big(f(\vec{x}_{-i},\vec{1}_{i})- f(\vec{x}_{-i},\vec{0}_{i})\Big)
    \operatorname{d}x_1\dots\operatorname{d}x_{i-1}\, \operatorname{d}x_{i+1}\dots\operatorname{d}x_n,
  $$  
  so that $\widetilde{\psi}_i(f)\ge \frac{\delta^{n-1}\varepsilon}{2}>0$, $\sum_{h=1}^n \widetilde{\psi}_h(f)>0$, 
  and $\sum_{h=1}^n \psi_h(f)=1$, i.e., $\psi$ is efficient.
  
  Let $i,j\in N$ be distinct symmetric agents and $\vec{x}\in[0,1]^n$ with $x_i=x_j$. Symmetry of $\widetilde{\psi}$ and $\psi$ 
  follows from $f(\vec{x}_{-i},\vec{1}_{i})-f(\vec{x}_{-i},\vec{0}_{i})=f(\vec{x}_{-j},\vec{1}_{j})-f(\vec{x}_{-j},\vec{0}_{j})$.   
  
  Since $f(\vec{x}_{-i},\vec{1}_{i})=f(\vec{x}_{-i},\vec{0}_{i})$ for every null $i\in N$, $\widetilde{\psi}$ 
  and $\psi$ satisfy the null property.
  
  Since $x+y=\max\{x,y\}+\min\{x,y\}$ for all $x,y\in\mathbb{R}$ and due to the linearity of finite sums and integrals, $\widetilde{\psi}$ 
  satisfies the transfer axiom.
\end{proof}

Next, we show that the examples $\hat{f}(x_1,x_2,x_3)=\frac{1x_1^2+2x_2^2+3x_3^2}{6}$ and $\tilde{f}(x_1,x_2,x_3)=x_1x_2^2x_3^3$ from 
Section~\ref{sec_importance_measures} both are part of more general families for which formulas for $\varphi$ and $\psi$ can be determined.

\begin{theorem}
  (Cf.\ \cite[Table  10.4]{grabischaggregation})
  \label{thm_linear_combination}
  For a positive integer $n$, $\vec{w}\in\mathbb{R}_{\ge 0}^n$ with $\sum_{i=1}^n w_i=1$, and surjective, continuous, 
  weakly monotonic increasing functions $f_i\colon[0,1]\to[0,1]$ for $1\le i\le n$, let $f\colon[0,1]^n\to[0,1]$ defined 
  by $\vec{x}\mapsto \sum_{i=1}^n w_i\cdot f_i(x_i)$. With this, $f$ is a simple aggregation function and 
  $\varphi_i(f)=\psi_i(f)=w_i$ for all $1\le i\le n$.
\end{theorem}
\begin{proof}
  By construction, $f$ is continuous. If $\vec{x}\le\vec{y}$ for two vectors $x,y\in[0,1]^n$, then $x_i\le y_i$ for all 
  $1\le i\le n$ so that $f_i(x_i)\le f_i(y_i)$ and $w_i\cdot f_i(x_i)\le w_i\cdot f_i(y_i)$, which implies that $f$ is 
  weakly monotonic increasing. For $1\le i\le n$ we define $\underline{\vec{x}}_i=(1,\dots,1,0,\dots,0)\in[0,1]^n$ with 
  $i-1$ leading ones and $\overline{\vec{x}}_i=(1,\dots,1,0,\dots,0)\in[0,1]^n$ with $i$ leading ones. Since the $f_i$ are 
  surjective, $f_i(0)=0$ and $f_i(1)=1$, the image of $\lambda\cdot \underline{\vec{x}}_i+(1-\lambda)\cdot \overline{\vec{x}}_i$ 
  for $\lambda\in[0,1]$ under $f$ is given by $\left[\sum_{j=1}^{i-1} w_j,\sum_{j=1}^{i} w_j\right]$ for all $1\le i\le n$. 
  Thus, $f$ is surjective.
  
  Since 
  \begin{equation}
    f(\vec{x}_{-i},\vec{1}_{i})-f(\vec{x}_{-i},\vec{0}_{i})=w_i\cdot f_i(1)-w_i\cdot f_i(0)=w_i
  \end{equation}
  for all $1\le i\le n$, we have $\widetilde{\psi}_i(f)=w_i$ and $\psi_i(f)=w_i$ (due to $\sum_{j=1}^n w_j=1$).
  
  For each $\pi\in \mathcal{S}_n$ and each $i\in N$ we similarly have
  \begin{equation}
    f(\upim)-f(\dnim)=
    \sum_{j=\pi^{-1}(i)}^{n} w_{\pi(j)}
  \end{equation}
  and
  \begin{equation}  
    f(\upi)-f(\dni)=
    \sum_{j=\pi^{-1}(i)+1}^{n} w_{\pi(j)}.
  \end{equation}  
  Thus, the difference equals $w_i$, so that $\varphi_i(f)=w_i$.
\end{proof}

\begin{theorem}
  (Cf.\ \cite[Table  10.4]{grabischaggregation})
  \label{thm_power_product}
  For a positive integer $n$ and positive real numbers $\alpha_1,\dots,\alpha_n$ let $f\colon[0,1]^n\to[0,1]$ be defined by 
  $\vec{x}\mapsto \prod_{i=1}^n x_i^{\alpha_i}$ and $\Lambda=\prod_{j=1}^n \left(\alpha_j+1\right)$. Then, $f$ is a simple aggregation function,
  \begin{equation}
    \varphi_i(f)=\frac{1}{n!\cdot\Lambda}\cdot\left( (n-1)!+\alpha_i\cdot \sum_{T\subseteq N\backslash\{i\}} |T|!\cdot(n-1-|T|)! \cdot 
    \prod_{j\in T} \left(a_j+1\right)\right) 
  \end{equation}
  for each agent $i\in N$, and $\psi_i(f)=\frac{\alpha_i+1}{n+\sum_{j=1}^n \alpha_j}$ for each agent $i\in N$ if $n\ge 2$.
\end{theorem}
\begin{proof}
  We directly check the conditions from Definition~\ref{def_simple}, cf.\ the proof of Theorem~\ref{thm_linear_combination}. For 
  each agent $i\in N$ employing the definition of $f$ gives
  \begin{equation}
    \widetilde{\psi}_i(f)=\int_0^1\dots\int_0^1 \prod_{j=1}^{i-1} x_j^{\alpha_j}\cdot \prod_{j=i+1}^{n} x_j^{\alpha_j}  
    \operatorname{d}x_1\dots\operatorname{d}x_{i-1}\, \operatorname{d}x_{i+1}\dots\operatorname{d}x_n.
  \end{equation}
  Since $\int_0^1 cx^\alpha\operatorname{d}x=\frac{c}{\alpha+1}$ for each $\alpha>0$, we recursively compute 
  \begin{equation}
    \widetilde{\psi}_i(f)=\frac{\alpha_i+1}{\prod\limits_{j=1}^n \left(\alpha_j+1\right)},
  \end{equation}  
  for each agent $i\in N$, so that the stated formula for $\psi$ follows.
  
  For the computation of $\varphi(f)$ we first observe 
  \begin{equation}
    \int_0^1\dots\int_0^1 \prod_{j=1}^r x_j^{\beta_j} \operatorname{d}x_1\,\dots\,\operatorname{d}x_n=
    \prod_{j=1}^r \frac{1}{\beta_j+1}    
  \end{equation}
  for $0\le r\le n$ and $\beta_j\in\mathbb{R}_{>0}$ for $1\le j\le r$. Now let $i\in N$ an arbitrary but fixed agent and 
  $\pi\in\mathcal{S}_n$ an arbitrary but fixed permutation. As abbreviation we set $S=\left\{j\in N\,:\, \pi(j)<\pi(i)\right\} 
  \subseteq N\backslash\{i\}$. With this, we have
  \begin{eqnarray}
    && \int_0^1\dots\int_0^1
    \Big(f(\upim)-f(\upi)\Big)+\Big(f(\dni)-f(\dnim)\Big) 
    \operatorname{d}x_1\,\dots\,\operatorname{d}x_n\nonumber\\
    &=& \left(\prod_{j\in S} \frac{1}{a_j+1}\right)\cdot \left(1-\frac{1}{\alpha_i+1}\right)=\frac{\alpha_i}{\Lambda}\cdot
        \prod_{j\in N\backslash (S\cup\{i\})} \left(a_j+1\right) 
  \end{eqnarray}
  for $S\neq N\backslash\{i\}$. For $S=N\backslash \{i\}$ the first expression 
  evaluates to $\frac{\alpha_i+1}{\Lambda}=\frac{\alpha_i}{\Lambda}+\frac{1}{\Lambda}$ instead of $\frac{\alpha_i}{\Lambda}$, so that 
  we obtain the stated formula for $\varphi_i(f)$.  
\end{proof}

In our second example in Section~\ref{sec_importance_measures} we have $n=3$, $\alpha_1=1$, $\alpha_2=2$, and $\alpha_3=3$, so that 
$n!\cdot \Lambda=144$. For $i=1$ the expression $|T|!\cdot(n-1-|T|)! \cdot \prod_{j\in T} \left(a_j+1\right)$ evaluates to 
$2$, $3$, $4$, and $24$ for $T=\emptyset$, $T=\{2\}$, $T=\{3\}$, and $T=\{2,3\}$, respectively, so that $\varphi_1(\tilde{f})=\frac{35}{144}$. 
For $i=2$ we obtain the values $2$, $2$, $4$, and $16$, so that $\varphi_2(\tilde{f})=\frac{50}{144}$.
 
\section{Weighted medians}
\label{sec_median} 

Here we want to give a definition for the weighted median as a simple aggregation function which is more general than 
the description from the introduction in Section~\ref{sec_introduction}. For a positive integer $n$ we consider 
real numbers $x_i$ for $1\le i\le n$. Arrange the values in weakly increasing order $x_{i_1}\le x_{i_2}\le \dots\le x_{i_n}$.  
For $\vec{w}\in\mathbb{R}_{\ge 0}^n\backslash\vec{0}$ let $1\le j\le n$ be the smallest index such that 
$\sum_{h=1}^j w_{i_h}\ge\sum_{h=1}^n w_h/2$. If we have equality, then we set the weighed median, with respect to $\vec{w}$, to 
$\left(x_{i_j}+x_{i_{j+1}}\right)/2$ and to $x_{i_j}$ otherwise. We can easily check that this definition gives a simple 
aggregation function, which we denote by $m_{\vec{w}}$.

In the introduction we have seen that different weight vectors can lead to the same function $m_{\vec{w}}$. Here we want to study the 
equivalence question and the correspondence to weighted games. We restrict ourselves onto the cases where 
$\vec{w}(S)\neq \vec{w}(N\backslash S)$ for every coalition $S\subseteq N$, i.e., there is always a unique weighted median. So, 
modeling as a simple game $v$, for any $S\subseteq N$ the two complementary sets $S$ and $N\backslash S$ form exactly one winning and 
exactly one losing coalition, i.e., $v$ is a weighted constant-sum game. Two simple aggregation functions $m_{\vec{w}}$ and 
$m_{\vec{w}'}$ coincide if and only if $[q;\vec{w}]=[q';\vec{w}']$, where $q=\vec{w}(N)/2$ and $q'=\vec{w}'(N)/2$.      

\begin{lemma}
  If $[q';\vec{w}']$ is constant-sum game for $n$ agents, then there exist $q\in\mathbb{N}_{>0}$ and $\vec{w}\in\mathbb{N}^n$ 
  such that $\sum_{i=1}^n w_i=2q-1$, which is odd, and $[q';\vec{w}']=[q;\vec{w}]$. 
\end{lemma}
\begin{proof}
  In Section~\ref{sec_binary_voting} we have argued that for every weighted game there exists a representation with integer weights and 
  integer quota. Choose an integer $q$ and $\vec{w}\in\mathbb{N}^n$ such that $[q';\vec{w}']=[q;\vec{w}]$ and $\sum_{i=1}^n w_i$ is minimized. 
  Let $l$ denote the maximum weight of a losing coalition and $u$ the minimum weight of a winning coalition, so that $l+1\le u$. If 
  $l+2\le u$, then we may decrease the weight of one agent, with a positive weight, by one and set the quota to $l+1$, which contradicts 
  the minimality of $\sum_{i=1}^n w_i$. Thus, we have $l+1=u=q$. Let $S$ be a winning coalition with weight $l+1$. Then $N\backslash S$ 
  is losing and has a weight of $\sum_{i=1}^n w_i -l-1\le l$, so that $\vec{w}(N)\le 2q-1$. Let $T$ be a losing coalition of weight $l$. 
  Then $N\backslash T$ is winning and has a weight of $\sum_{i=1}^n w_i -l\ge l+1$, so that $\vec{w}(N)\ge 2q-1$, which gives $\vec{w}(N)
  = 2q-1$.  
\end{proof}

We remark that for $n\ge 8$ agents there may be several representations of a weighted game with integer weights, 
integer quota, and minimum weight sum, see e.g.\ \cite{kurz2012minimum}.

\begin{proposition}
  \label{prop_bzi_median}
  Let $\vec{w}\in\mathbb{N}^n$, where $n\ge 3$ and $\sum_{i=1}^n w_i$ is odd. For $q=(\vec{w}(N)+1)/2$ let $f=m_{\vec{w}}$ be 
  the weighted median simple aggregation function corresponding to the weighted game $v=[q;\vec{w}]$. If all coalitions of size $1$ 
  are losing and all coalitions of size $n-1$ are winning in $v$, then 
  \begin{eqnarray}
    \widetilde{\psi}_i(f)&=&\sum_{S\subseteq N\backslash\{i\}\,:\, v(S)=1}
    \!\!\!\!\!\!\!\!\left|\left\{j\in S\,:\,v(S\backslash\{j\})=0\right\}\right|\cdot |S|!\cdot(n-1-|S|)!/n!\nonumber\\
    &&- \!\!\!\!\!\sum_{S\subseteq N\backslash\{i\}\,:\,v(S\cup\{i\})=1}
    \!\!\!\!\!\!\!\!\!\!\!\!\!\!\!\!\!\!\left|\left\{j\in S\,:\,v(S\backslash\{j\}\!\cup\!\{i\})=0\right\}\right|\cdot 
    |S|!\cdot(n\!-\!1\!-\!|S|)!/n!\label{formula_weighted_median_bzi}
  \end{eqnarray}
  for all $1\le i\le n$.
\end{proposition}
\begin{proof}  
We consider 
\begin{equation}
  \int_0^1\dots\int_0^1 f(\vec{x}_{-i},\vec{1}_{i})
  \operatorname{d}x_1\dots\operatorname{d}x_{i-1}\, \operatorname{d}x_{i+1}\dots\operatorname{d}x_n
\end{equation}
for an agent $i\in N$. In order to evaluate this expression we decompose the integration domain. For a given state vector $\vec{x}$ 
with pairwise different coordinates let $x_{i_1}<\dots<x_{i_n}=1$ be the ordering of the coordinates and $S\subseteq N\backslash\{i\}$ 
be a set of the form $\left\{i_h\,:\,1\le h\le j\right\}$, where $j$ is chosen minimal such that $S$ is winning in $[q;\vec{w}]$, i.e., 
$f(x_1,\dots,x_{i-1},1,x_{i+1},\dots,x_n)=x_{i_j}$. Note that $\emptyset\neq S\subseteq N\backslash\{i\}$ is a winning coalition and 
$S\backslash\{j\}$ is losing. If $s$ denotes the cardinality of $S$, then  
\begin{eqnarray}
  &&c(S)) \cdot \int_0^1\int_0^{x_s}\dots\int_0^{x_s} \int_{x_s}^1\dots\int_{x_s}^1 x_s \operatorname{d} x_{s+1}\dots \operatorname{d} x_{n-1}
  \operatorname{d} x_1\dots \operatorname{d} x_{s} \nonumber\\
  &=& c(S)\cdot \int_0^1 x^s(1-x)^{n-1-s}\operatorname{d}x=c(S)\cdot \frac{s!\cdot(n-1-s)!}{n!}
\end{eqnarray}
gives the value of the above integral over the integration domain that corresponds to $S$, where 
$c(S)=\left|\left\{j\in S\,:\,v(S\backslash\{j\})=0\right\}\right|$ denotes the number of \emph{critical} agents in $S$.

For the integral over $f(\vec{x}_{-i},\vec{0}_{i})$ only slight adjustments are necessary, so that we finally end up 
with the stated formula.
\end{proof}

For the importance measure $\varphi$ we can also eliminate the integrals in the definition of $\varphi$ for a simple aggregation function 
based on the weighted median. For each fixed agent $i\in N$ and each fixed permutation $\pi\in\mathcal{S}_n$ we define the set 
$S=\left\{j\in N\,:\,\pi(j)<\pi(i)\right\}\subseteq N\backslash\{i\}$. The value of 
$$
  f(\upim)-f(\upi)+f(\dni)-f(\dnim)
$$
is the same for any two permutations $\pi$ and $\pi'$, which correspond to the same set and there correspond exactly $|S|!\cdot(n-1-|S|)!$ 
permutations to a set $S\subseteq N\backslash\{i\}$. This allows to replace the sum over $\mathcal{S}_n$ by a sum over the subsets of 
$N\backslash\{i\}$. In order to evaluate the above expression, we need to know the ordering of the $x_j$ for all $j\in S$ or all $j\in S\cup\{i\}$. 
The value $x_h$ for an agent $h$ outside of such a set is not relevant since it is replaced by either a $0$ or a $1$ in the formula for $\varphi_i(f)$. 
  
For order $y_1\le y_2\le \dots\le y_r$ and weighted median $y_i$ we have
\begin{equation}\label{eq_int_y_i_in_ordering}
  \int_{0}^1 \int_{0}^{y_r}  \int_{0}^{y_{r-1}}\dots \int_{0}^{y_2} y_i \operatorname{d}y_1 \operatorname{d}y_2\dots \operatorname{d}y_r
  =\frac{i}{(r+1)!}.
\end{equation} 
If the weighted median is $1$, then we have
\begin{equation}
  \int_{0}^1 \int_{0}^{y_r}  \int_{0}^{y_{r-1}}\dots \int_{0}^{y_2} 1 \operatorname{d}y_1 \operatorname{d}y_2\dots \operatorname{d}y_r
  =\frac{1}{r!}.
\end{equation}  
It may also happen that the weighted median is zero, where the corresponding integral of course is zero.

So, we have implemented the following algorithm: We loop over all $S\subseteq N\backslash\{i\}$. Then, we loop over all 
possible orderings of the $x_j$, where $j\in S$. 
The integral over
$$
  \Big(f(\upim)-f(\dnim)\Big)
$$
for the integration domain according to the fixed ordering within $S$ can be evaluated by one of the above three cases - depending 
on the position of the weighted median in the ordering. Similarly, for 
$$
  \Big(f(\dni)-f(\upi)\Big)
$$ 
we loop over all possible orderings of the $x_j$, where $j\in S\cup\{i\}$, determine the position of the weighted median, and evaluate 
the integral. 

We have applied the algorithm sketched above for all weighted constant-sum games with up to $n=9$ agents, which supports:
\setcounter{theorem}{0}
\begin{conjecture} 
  \label{conj_ssi_median}
  Let $\vec{w}\in\mathbb{N}^n$, where $n\ge 1$ and $\sum_{i=1}^n w_i$ is odd. For $q=(\vec{w}(N)+1)/2$ let $f=m_{\vec{w}}$ be 
  the weighted median simple aggregation function corresponding to the weighted game $v=[q;\vec{w}]$. Then, 
  $\varphi(f)$ coincides with the Shapley-Shubik index of $[q;\vec{w}]$.
\end{conjecture} 

We remark that the number of weighted constant-sum games with up to $n=9$ agents is given by $1$, $1$, $2$, $3$, $7$, $21$, $135$, 
$2470$, and $175\,428$, respectively, see also Table~1 in \cite{0841.90134}, where those objects were called {\lq\lq}games in $Z_n^r${\rq\rq}.

\section{Conclusion}
\label{sec_conclusion}

We have introduced simple aggregation functions which mimic simple games. Via importance measures we aim to measure the importance 
of an agent in a given simple aggregation function. As outlined in the introduction, there is a large variety of applications for this 
quantification. Exemplarily, we have introduced two importance measures, which mimic the Shapley-Shubik and the Penrose-Banzhaf index, 
respectively. Having proven several properties of these two importance measures, a suitable axiomatization remains an open problem. The evaluation 
of both importance measures is computationally involved. For two parametric classes of simple aggregation functions we have derived 
a more direct formula. For simple aggregation functions based on the weighted median, we have stated a reasonable simplification for 
$\psi$. It would be interesting to study whether the expression from Proposition~\ref{prop_bzi_median} can serve as a reasonable power 
index for simple games. We also stated a first simplification for $\varphi$. It would be quite interesting to know if 
Conjecture~\ref{conj_ssi_median} is true in general.

A further line of research might be to identify other interesting parametric classes of simple aggregation functions (arising from 
applications), see \cite{grabischaggregation}. In a second step, exact formulas for our two importance measures are beneficial. Also the 
generalization of further 
power indices for simple games to importance measures or the development of completely new importance measures is an interesting task for the 
future. To incorporate social influence between agents seems worthwhile to study, since social influence can, e.g., undermine the wisdom 
of the crowd effect, see \cite{lorenz2011social}.    


\begin{thebibliography}{10}
\providecommand{\urlprefix}{}
\expandafter\ifx\csname urlstyle\endcsname\relax
  \providecommand{\doi}[1]{doi:\discretionary{}{}{}#1}\else
  \providecommand{\doi}{doi:\discretionary{}{}{}\begingroup
  \urlstyle{rm}\Url}\fi

\bibitem{bahng2012relationship}
Bahng, Y. and Kincade, D.~H., The relationship between temperature and sales:
  Sales data analysis of a retailer of branded women's business wear,
  \emph{International Journal of Retail \& Distribution Management} \textbf{40}
  (2012) 410--426.

\bibitem{banzhaf1964weighted}
Banzhaf~III, J.~F., Weighted voting doesn't work: A mathematical analysis,
  \emph{Rutgers Law Review} \textbf{19} (1964) 317--343.

\bibitem{chen2004eliminating}
Chen, K.-Y., Fine, L.~R., and Huberman, B.~A., Eliminating public knowledge
  biases in information-aggregation mechanisms, \emph{Management Science}
  \textbf{50} (2004) 983--994.

\bibitem{clark1976effects}
Clark, W. A.~V. and Avery, K.~L., The effects of data aggregation in
  statistical analysis, \emph{Geographical Analysis} \textbf{8} (1976)
  428--438.

\bibitem{clemen1999combining}
Clemen, R.~T. and Winkler, R.~L., Combining probability distributions from
  experts in risk analysis, \emph{Risk Analysis} \textbf{19} (1999) 187--203.

\bibitem{cooper2009handbook}
Cooper, H., Hedges, L.~V., and Valentine, J.~C., \emph{The handbook of research
  synthesis and meta-analysis} (Russell Sage Foundation, 2009).

\bibitem{de2017inverse}
De, A., Diakonikolas, I., and Servedio, R.~A., The inverse {S}hapley value
  problem, \emph{Games and Economic Behavior} \textbf{105} (2017) 122--147.

\bibitem{delucchi2011providing}
Delucchi, M.~A. and Jacobson, M.~Z., Providing all global energy with wind,
  water, and solar power, part ii: Reliability, system and transmission costs,
  and policies, \emph{Energy Policy} \textbf{39} (2011) 1170--1190.

\bibitem{dragan2005inverse}
Dragan, I.~C., On the inverse problem for semivalues of cooperative {T}{U}
  games, \emph{International Journal of Pure and Applied Mathematics}
  \textbf{22} (2005) 539--555.

\bibitem{dubey1975uniqueness}
Dubey, P., On the uniqueness of the {S}hapley value, \emph{International
  Journal of Game Theory} \textbf{4} (1975) 131--139.

\bibitem{dubey1979mathematical}
Dubey, P. and Shapley, L., Mathematical properties of the {B}anzhaf power
  index, \emph{Mathematics of Operations Research} \textbf{4} (1979) 99--131.

\bibitem{faigle2016bases}
Faigle, U. and Grabisch, M., Bases and linear transforms of {T}{U}-games and
  cooperation systems, \emph{International Journal of Game Theory} \textbf{45}
  (2016) 875--892.

\bibitem{felsenthal1996alternative}
Felsenthal, D. and Machover, M., Alternative forms of the {S}hapley value and
  the {S}hapley-{S}hubik index, \emph{Public Choice} \textbf{87} (1996)
  315--318.

\bibitem{felsenthal1998measurement}
Felsenthal, D.~S., Machover, M., \emph{et~al.}, \emph{The measurement of voting
  power} (Edward Elgar Publishing, Cheltenham, 1998).

\bibitem{fischer2002weighted}
Fischer, M., Paredes, J.~L., and Arce, G.~R., Weighted median image sharpeners
  for the world wide web, \emph{IEEE Transactions on Image Processing}
  \textbf{11} (2002) 717--727.

\bibitem{galton1907ballot}
Galton, F., The ballot-box, \emph{Nature} \textbf{75} (1907) 509.

\bibitem{galton1907vox}
Galton, F., Vox populi, \emph{Nature} \textbf{75} (1907) 450--451.

\bibitem{genest1986combining}
Genest, C. and Zidek, J.~V., Combining probability distributions: A critique
  and an annotated bibliography, \emph{Statistical Science}  (1986) 114--135.

\bibitem{grabischaggregation}
Grabisch, M., Marichal, J.-L., Mesiar, R., and Pap, E., \emph{Aggregation
  Functions} (Cambridge Univ. Press, Cambridge, 2009).

\bibitem{hegselmann2002opinion}
Hegselmann, R. and Krause, U., Opinion dynamics and bounded confidence models,
  analysis, and simulation, \emph{Journal of Artificial Societies and Social
  Simulation} \textbf{5} (2002).

\bibitem{hegselmann2005opinion}
Hegselmann, R. and Krause, U., Opinion dynamics driven by various ways of
  averaging, \emph{Computational Economics} \textbf{25} (2005) 381--405.

\bibitem{holcombe1989median}
Holcombe, R.~G., The median voter model in public choice theory, \emph{Public
  Choice} \textbf{61} (1989) 115--125.

\bibitem{hu2006asymmetric}
Hu, X., An asymmetric {S}hapley--{S}hubik power index, \emph{International
  Journal of Game Theory} \textbf{34} (2006) 229--240.

\bibitem{kanazawa1998brief}
Kanazawa, S., A brief note on a further refinement of the {C}ondorcet {J}ury
  {T}heorem for heterogeneous groups, \emph{Mathematical Social Sciences}
  \textbf{35} (1998) 69--73.

\bibitem{koriyama2013optimal}
Koriyama, Y., Mac{\'e}, A., Treibich, R., and Laslier, J.-F., Optimal
  apportionment, \emph{Journal of Political Economy} \textbf{121} (2013)
  584--608.

\bibitem{krishnamachari2002impact}
Krishnamachari, L., Estrin, D., and Wicker, S., The impact of data aggregation
  in wireless sensor networks, in \emph{Distributed Computing Systems
  Workshops, 2002. Proceedings. 22nd International Conference on} (IEEE, 2002),
  pp. 575--578.

\bibitem{0841.90134}
Krohn, I. and Sudh\"olter, P., Directed and weighted majority games.,
  \emph{Mathematical Methods of Operations Research} \textbf{42} (1995)
  189--216.

\bibitem{kurz2012minimum}
Kurz, S., On minimum sum representations for weighted voting games,
  \emph{Annals of Operations Research} \textbf{196} (2012) 361--369.

\bibitem{kurz2012inverse}
Kurz, S., On the inverse power index problem, \emph{Optimization} \textbf{61}
  (2012) 989--1011.

\bibitem{kurz2014measuring}
Kurz, S., Measuring voting power in convex policy spaces, \emph{Economies}
  \textbf{2} (2014) 45--77.

\bibitem{kurz2017democratic}
Kurz, S., Maaser, N., and Napel, S., On the democratic weights of nations,
  \emph{Journal of Political Economy} \textbf{125} (2017) 1599--1634.

\bibitem{lorenz2011social}
Lorenz, J., Rauhut, H., Schweitzer, F., and Helbing, D., How social influence
  can undermine the wisdom of crowd effect, \emph{Proceedings of the National
  Academy of Sciences} \textbf{108} (2011) 9020--9025.

\bibitem{maaser2007equal}
Maaser, N. and Napel, S., Equal representation in two-tier voting systems,
  \emph{Social Choice and Welfare} \textbf{28} (2007) 401--420.

\bibitem{martin2017owen}
Martin, M., Nganmeni, Z., and Tchantcho, B., The {O}wen and {S}hapley spatial
  power indices: a comparison and a generalization, \emph{Mathematical Social
  Sciences} \textbf{89} (2017) 10--19.

\bibitem{mcmurray2012aggregating}
McMurray, J.~C., Aggregating information by voting: The wisdom of the experts
  versus the wisdom of the masses, \emph{Review of Economic Studies}
  \textbf{80} (2012) 277--312.

\bibitem{MayerBooker1991}
Meyer, M.~A. and Booker, J.~M., \emph{Eliciting and Analyzing Expert Judgment:
  A Practical Guide} (SIAM, Philadelphia, 2001).

\bibitem{owen1972multilinear}
Owen, G., Multilinear extensions of games, \emph{Management Science}
  \textbf{18} (1972) 64--79.

\bibitem{penrose1946elementary}
Penrose, L.~S., The elementary statistics of majority voting, \emph{Journal of
  the Royal Statistical Society} \textbf{109} (1946) 53--57.

\bibitem{rauhut2011wisdom}
Rauhut, H. and Lorenz, J., The wisdom of crowds in one mind: How individuals
  can simulate the knowledge of diverse societies to reach better decisions,
  \emph{Journal of Mathematical Psychology} \textbf{55} (2011) 191--197.

\bibitem{riker1986first}
Riker, W.~H., The first power index, \emph{Soc. Choice Welf.} \textbf{3} (1986)
  293--295.

\bibitem{shapley1954method}
Shapley, L. and Shubik, M., A method for evaluating the distribution of power
  in a committee system, \emph{American Political Science Review} \textbf{48}
  (1954) 787--792.

\bibitem{surowiecki2004wisdom}
Surowiecki, J., \emph{The wisdom of crowds: Why the many are smarter than the
  few and how collective wisdom shapes business, Economies, Societies and
  Nations}, Vol. 296 (Anchor Books, New York, 2005).

\bibitem{von1953theory}
von Neumann, J. and Morgenstern, O., \emph{Theory of Games and Economic
  Behavior}, 3rd edn. (Princeton University Press, 1953).

\bibitem{xi2009statistical}
Xi, R., \emph{Statistical aggregation: Theory and applications}, Ph.D. thesis,
  Washington University in St. Louis (2009).

\end{thebibliography}

\appendix

\section{An example for the weighted median for $\psi$}

Consider the simple aggregation function given by the weighted median based on the weighted game $[3;2,1,1,1]$ for four agents. 
We aim to evaluate $\psi(f)$. We start by computing $\widetilde{\psi}_1(f)$. Assuming $x_2\le x_3\le x_4$ we have 
$f(1,x_2,x_3,x_4)=x_4$ and $f(0,x_2,x_3,x_4)=x_2$, so that
$$
  \int_{0}^1 \int_{0}^{x_4} \int_{0}^{x_3} f(1,x_2,x_3,x_4)\operatorname{d}x_2\operatorname{d}x_3\operatorname{d}x_4
  =\int_{0}^1 \int_{0}^{x_4} \int_{0}^{x_3} x_4\operatorname{d}x_2\operatorname{d}x_3\operatorname{d}x_4
  =\frac{1}{8}
$$  
and 
$$
  \int_{0}^1 \int_{0}^{x_4} \int_{0}^{x_3} f(0,x_2,x_3,x_4)\operatorname{d}x_2\operatorname{d}x_3\operatorname{d}x_4
  =\int_{0}^1 \int_{0}^{x_4} \int_{0}^{x_3} x_2\operatorname{d}x_2\operatorname{d}x_3\operatorname{d}x_4
  =\frac{1}{24}.
$$
We remark that Equation~(\ref{eq_int_y_i_in_ordering}) gives a general formula for the above integrals. 
Since agents $2$, $3$, and $4$ are symmetric, we have the same result for all $3!$ orderings of $x_2$, $x_3$, and $x_4$, so that 
$\widetilde{\psi}_1(f)=\frac{3}{4}-\frac{1}{4}=\frac{1}{2}$. Note that for the first sum in Equation~(\ref{formula_weighted_median_bzi}) 
only $S=\{2,3,4\}$ leads to a non-zero multiplier, which here is $3$. In the second sum only the set $\{2\}$, $\{3\}$, and $\{4\}$ 
lead to a non-zero multiplier, which always is $1$. Of course this gives the same result. 

Next we compute $\widetilde{\psi}_2(f)$. If $x_1\le x_3\le x_4$, then $f(x_1,1,x_3,x_4)=x_3$ and $f(x_1,0,x_3,x_4)=x_1$, so that 
$$
  \int_{0}^1 \int_{0}^{x_4} \int_{0}^{x_3} f(x_1,1,x_3,x_4)\operatorname{d}x_1\operatorname{d}x_3\operatorname{d}x_4
  =\int_{0}^1 \int_{0}^{x_4} \int_{0}^{x_3} x_3\operatorname{d}x_1\operatorname{d}x_3\operatorname{d}x_4
  =\frac{1}{12}
$$
and
$$
  \int_{0}^1 \int_{0}^{x_4} \int_{0}^{x_3} f(x_1,0,x_3,x_4)\operatorname{d}x_1\operatorname{d}x_3\operatorname{d}x_4
  =\int_{0}^1 \int_{0}^{x_4} \int_{0}^{x_3} x_1\operatorname{d}x_1\operatorname{d}x_3\operatorname{d}x_4
  =\frac{1}{24}.
$$
If $x_3\le x_1\le x_4$, then $f(x_1,1,x_3,x_4)=x_1$ and $f(x_1,0,x_3,x_4)=x_1$, so that
$$
  \int_{0}^1 \int_{0}^{x_4} \int_{0}^{x_1} f(x_1,1,x_3,x_4)\operatorname{d}x_3\operatorname{d}x_1\operatorname{d}x_4
  =\int_{0}^1 \int_{0}^{x_4} \int_{0}^{x_1} x_1\operatorname{d}x_3\operatorname{d}x_1\operatorname{d}x_4
  =\frac{1}{12}
$$
and
$$
  \int_{0}^1 \int_{0}^{x_4} \int_{0}^{x_1} f(x_1,0,x_3,x_4)\operatorname{d}x_3\operatorname{d}x_1\operatorname{d}x_4
  =\int_{0}^1 \int_{0}^{x_4} \int_{0}^{x_1} x_1\operatorname{d}x_3\operatorname{d}x_1\operatorname{d}x_4.
  =\frac{1}{12}.
$$
If $x_3\le x_4\le x_1$, then $f(x_1,1,x_3,x_4)=x_1$ and $f(x_1,0,x_3,x_4)=x_4$, so that
$$
  \int_{0}^1 \int_{0}^{x_1} \int_{0}^{x_4} f(x_1,1,x_3,x_4)\operatorname{d}x_3\operatorname{d}x_4\operatorname{d}x_1
  =\int_{0}^1 \int_{0}^{x_1} \int_{0}^{x_4} x_1\operatorname{d}x_3\operatorname{d}x_4\operatorname{d}x_1
  =\frac{1}{8}
$$
and
$$
  \int_{0}^1 \int_{0}^{x_1} \int_{0}^{x_4} f(x_1,0,x_3,x_4)\operatorname{d}x_3\operatorname{d}x_4\operatorname{d}x_1
  =\int_{0}^1 \int_{0}^{x_1} \int_{0}^{x_4} x_4\operatorname{d}x_3\operatorname{d}x_4\operatorname{d}x_1
  =\frac{1}{12}.
$$
Exchanging agent~$3$ with agent~$4$ gives the same number again, so that 
$$
  \widetilde{\psi}_2(f)=2\cdot\left(\frac{1}{12}+\frac{1}{12}+\frac{1}{8}\right)-2\cdot\left(\frac{1}{24}+\frac{1}{12}+\frac{1}{12}\right)
  =\frac{7}{12}-\frac{5}{12}=\frac{1}{6}.
$$
Looking at the first sum of Equation~(\ref{formula_weighted_median_bzi}) again, we have multipliers of two for $S=\{1,3\}$ or $S=\{1,4\}$, a
multiplier of one for $S=\{1,3,4\}$, and multipliers of zero in all other cases. I.e., the first sum equals $\frac{2\cdot 2\cdot 1}{24}+
\frac{2\cdot 2\cdot 1}{24}+\frac{1\cdot 6\cdot 1}{24}=\frac{7}{12}$. For the second sum we have a multiplier of two for $S=\{3,4\}$, multipliers 
of one for $S=\{1\}$, $S=\{1,3\}$, $S=\{1,4\}$, and multipliers of zero in all other cases. I.e., the second sum equals $\frac{2\cdot 2\cdot 1}{24}+
\frac{1\cdot 1\cdot 2}{24}+\frac{1\cdot 2\cdot 1}{24}+\frac{1\cdot 2\cdot 1}{24}=\frac{5}{12}$. 

The resulting power distribution $\frac{1}{6}\cdot(3,1,1,1)$ coincides with the Shapley-Shubik, the Penrose-Banzhaf, and the Public Good 
index of $[3;2,1,1,1]$. However, for $[4;3,1,1,1,1]$ we should get a power distribution of $\frac{1}{32}\cdot(12,5,5,5,5)$, which differs 
from the three power indices of the considered weighted game. If no error occurred during the computation, then for $[4;3,2,1,1]$ we get 
$\tilde{\psi}_1(f)=\frac{18}{120}-\frac{6}{120}=\frac{1}{10}$, $\tilde{\psi}_2(f)=\frac{14}{120}-\frac{10}{120}=\frac{1}{30}$, and 
$\tilde{\psi}_3(f)=\tilde{\psi}_4(f)=\frac{14}{120}-\frac{6}{120}=\frac{1}{15}$. Note that the power distribution is not monotone in the 
weights.

\end{document}